\documentclass[letterpaper,11pt]{article}
\usepackage[margin=1in]{geometry}
\usepackage{amsfonts,amsmath,amssymb,amsthm}
\usepackage{thmtools} \usepackage{thm-restate}
\usepackage[noadjust]{cite}
\usepackage{comment,enumerate,hyperref}
\usepackage[colorinlistoftodos,textsize=small,color=red!25!white,obeyFinal]{todonotes}
\newtheorem{theorem}{Theorem}[section]
\newtheorem{corollary}[theorem]{Corollary}

\newtheorem{lemma}[theorem]{Lemma}
\newtheorem{claim}[theorem]{Claim}

\newtheorem{oq}[theorem]{Open Question}

\usepackage{soul}
\theoremstyle{definition}
\newtheorem{definition}{Definition}

\usepackage{xcolor,xspace}
\usepackage{booktabs}

\usepackage[boxruled, noend]{algorithm2e}
\usepackage{placeins}

\usepackage{multirow}

\usetikzlibrary{decorations.pathmorphing}

\tikzset{snake it/.style={decorate, decoration=snake}}

\usepackage{csquotes}

\newcommand{\dist}{\text{dist}}
\newcommand{\Oish}{\widetilde{O}}

\newcommand{\Thetaish}{\widetilde{\Theta}}

\title{Epic Fail: Emulators can tolerate polynomially many edge faults for free}

\author{Greg Bodwin\thanks{Supported in part by NSF award CCF-2153680}\\ University of Michigan\\ bodwin@umich.edu \and
Michael Dinitz\thanks{Supported in part by NSF award CCF-1909111.} \\ Johns Hopkins University\\ mdinitz@cs.jhu.edu\and
Yasamin Nazari \thanks{Supported in part by Austrian Science Fund (FWF) grant P 32863-N.}\\ University of Salzburg\\ ynazari@cs.sbg.ac.at}

\date{}

\begin{document}
\maketitle

\begin{abstract} 
A $t$-emulator of a graph $G$ is a graph $H$ that approximates its pairwise shortest path distances up to multiplicative $t$ error.
We study fault tolerant $t$-emulators, under the model recently introduced by Bodwin, Dinitz, and Nazari [ITCS 2022]  for \emph{vertex failures}.
In this paper we consider the version for \emph{edge failures}, and show that they exhibit surprisingly different behavior.

In particular, our main result is that, for $(2k-1)$-emulators with $k$ odd, we can tolerate a polynomial number of edge faults \emph{for free}.
For example: for any $n$-node input graph, we construct a $5$-emulator ($k=3$) on $O(n^{4/3})$ edges that is robust to $f = O(n^{2/9})$ edge faults.
It is well known that $\Omega(n^{4/3})$ edges are necessary even if the $5$-emulator does not need to tolerate \emph{any} faults. 
Thus we pay no extra cost in the size to gain this fault tolerance.
We leave open the precise range of free fault tolerance for odd $k$, and whether a similar phenomenon can be proved for even $k$.

\end{abstract}

\renewcommand{\arraystretch}{1.5} 

\section{Introduction}

A basic question in theoretical computer science is to sparsify an input graph $G$ while approximately preserving its shortest path distances.
These sparsifiers have applications in networking, algorithms, and more (see survey \cite{ahmed2020graph}). We will discuss two important types of sparsifiers: spanners and emulators.  

\begin{definition} [Spanners]
For an input graph $G = (V, E, w)$, a $t$-spanner is a subgraph $H$ such that for all $u, v \in V$, we have
$\dist_G(u, v) \le \dist_H(u, v) \le t \cdot \dist_G(u, v)$.
The parameter $t$ is called the stretch of the spanner. 
\end{definition}

\begin{definition} [Emulators] \label{def:emulators}
For an input graph $G = (V, E, w)$, a $t$-emulator is a graph $H = (V, E_H)$ with the following property.
If we set the weight of each edge $(u, v) \in E_H$ to be $\dist_G(u, v)$, then for all $u, v \in V$, we have
$\dist_G(u, v) \le \dist_H(u, v) \le t \cdot \dist_G(u, v).$
\end{definition}

Elsewhere in the literature, emulators are often defined slightly differently, as arbitrary weighted graphs $H = (V, E_H, w)$ that approximate distances of $G$.
However, once we choose the edge set $E_H$, the choice of edge weights is clear: for an emulator edge $(u, v)$, we must choose $w(u, v) \ge \dist_G(u, v)$ to satisfy the first inequality in Definition \ref{def:emulators}, and there is no advantage to choosing $w(u, v)$ strictly larger than $\dist_G(u, v)$.\footnote{We mean here that there is no advantage in terms of quality of distance approximation.  One could imagine an advantage to choosing larger weights for other reasons, e.g., perhaps we want to compute the emulator quickly and so we use a lossy approximation algorithm to set weights.  But these concerns will not arise in this paper.}
Hence, this optimal choice of edge weights is hardcoded into Definition \ref{def:emulators}.
(This hardcoding will be convenient in our discussion of fault tolerant emulators below.)
We also note that, for both spanners and emulators, the first inequality $\dist_G(u, v) \le \dist_H(u, v)$ is trivially implied by the fact that $H$ is a subgraph of $G$ (for spanners) or the way the weights are set (for emulators).

Spanners were first introduced by~\cite{PelegS:89,PelegU:89} in the context of distributed computing, and they have since found many uses throughout algorithms.

Emulators were first introduced in the context of shortest path algorithms \cite{dor2000all}, but they have close connections to distributed computing as well, e.g.\ they form a reasonable model of overlay networks \cite{BDR22}.
A priori, one could imagine that the additional flexibility of emulators allows for a better size/stretch tradeoff than the one available to spanners.
But this possibility was refuted in a classic 1993 work of Alth{\" o}fer, Das, Dobkin, Joseph, and Soares, which determined the same (up to logs) tradeoff for both spanners and emulators:
\begin{theorem} [\cite{AlthoferDDJS:93}] \label{thm:nonftspanners} \label{thm:tightspanners}
For any positive integer $k$, every $n$-node graph has a $(2k-1)$-\textbf{spanner} on $O(n^{1+1/k})$ edges.
Moreover, assuming the girth conjecture~\cite{erdHos1964extremal}, this is essentially best possible in the sense that there are $n$-node graphs that do not admit $(2k-1)$-\textbf{emulators} on $o(n^{1+1/k} / \log n)$ edges.
\end{theorem}

\subsection{Fault-Tolerance}

Distributed computing often deals with networks in which vertices or edges can temporarily fail.
Thus, when spanners and emulators are applied in distributed contexts, it is natural to ask for \emph{fault-tolerant} versions of these sparsifiers.  The precise definitions are deferred to Section~\ref{sec:FTdefs}, but at a high level we say that a spanner or emulator $H$ is $f$-fault tolerant if, for every set of at most $f$ failing parts (edges or vertices), what remains of $H$ is still a spanner/emulator of what remains of $G$.
This leads to a clear question: what is the ``price'' of fault-tolerance?
That is, if we want to build an $f$-fault-tolerant spanner or emulator with stretch $2k-1$, what factor (depending on $f$) must we pay in the sparsity bound, above the baseline size of $O(n^{1+1/k})$ determined by Theorem \ref{thm:nonftspanners} for the non-fault tolerant version?

There has been a fruitful line of work tackling this question for spanners and emulators (e.g., \cite{ChechikLPR:10,DinitzK:11,BDPW18,BP19,DR20,BDR21,BDR22,BDN22,Par22}).
This has produced optimal bounds on the price of fault tolerance in some settings, distinctions between edge and vertex faults, and a suite of algorithmic techniques that allow for efficient computation of sparse fault tolerant spanners/emulators (including in distributed and parallel models of computation).
In particular, following~\cite{ChechikLPR:10,DinitzK:11}, results in \cite{BDPW18, BDR22, BDN22} have established that for edge/vertex fault tolerant spanners and for vertex fault tolerant emulators, the price of fault tolerance is a sublinear polynomial in $f$.
Such a dependence can still be expensive for the important scenario where $f$ is large, i.e.~a (sublinear) polynomial in the number of nodes $n$.

This paper is the first to investigate the remaining setting of edge fault tolerant emulators, and we show a quantitative size bound that is surprisingly different from the ones for these related objects.
Our main result is that there is a parameter regime $f \le \text{poly}(n)$ in which the price of edge fault tolerance for emulators is \emph{free}.
That is, for fixed odd $k$, we construct $(2k-1)$-emulators that can handle $f$ edge faults and only have $O(n^{1+1/k})$ edges, the same size bound settled by Theorem \ref{thm:tightspanners} for the non-faulty setting.

\subsubsection{Fault-Tolerance Definitions and Previous Results} \label{sec:FTdefs}
We now define the objects that we study and reference.  For general graph spanners, fault-tolerance was initially studied by Chechik, Langberg, Peleg, and Roditty~\cite{ChechikLPR:10}, who introduced the following definition:
\begin{definition} [EFT Spanners] \label{def:eftspan}
For an input graph $G$, an $f$-edge fault tolerant (EFT) $t$-spanner is a subgraph $H$ of $G$ such that $\dist_{G \setminus F}(u,v) \leq \dist_{H \setminus F}(u,v) \leq t \cdot \dist_{G \setminus F}(u,v)$ for all $F \subseteq E$ with $|F| \leq f$ and for all $u,v \in V$. 
\end{definition}
Equivalently, $H \setminus F$ must be a $t$-spanner of $G \setminus F$ for all possible fault sets $F$.
Vertex fault-tolerance is defined analogously, with $F$ as a vertex set rather than an edge set.

After significant work following~\cite{ChechikLPR:10} (see in particular~\cite{BDPW18,BP19,DR20,BDR21,BDR22}), we now have a generally good understanding of both vertex- and edge-fault tolerant spanners.  Precise bounds are given in Table~\ref{tbl:ftprice}, but the high-level view is that the price of fault-tolerance is $f^{1-1/k}$ for vertex fault tolerant spanners, and even smaller for edge fault tolerant spanners (approximately $f^{1/2}$, although the exact bounds remain open).
With fault-tolerant spanners relatively well understood, fault-tolerant \emph{emulators} were first studied by Bodwin, Dinitz, and Nazari \cite{BDN22}.
They studied emulators under vertex failures, but their model under edge failures is the following:

\begin{definition}[EFT Emulators]
For an input graph $G = (V, E)$, we say that $H = (V, E')$ is an $f$-edge fault tolerant (EFT) $t$-emulator $H$ if, for all $F \subseteq E$ with $|F| \le f$, $H$ is a $t$-emulator of $G \setminus F$.
\end{definition}

Recall that Definition \ref{def:emulators} hardcodes the edge weights of emulators, and thus this definition specifically means that $H$ approximates the distances of $G \setminus F$ \emph{when the edge weights of $H$ are set according to distances in $G \setminus F$}.
That is, the following definition of EFT emulators is equivalent to the above.
Given a set $F \subseteq E$, let $H^F = (V, E', w_F)$ be the weighted version of $H$ where the weight of each edge $(u,v) \in E'$ is set to $w_F(u,v) = \dist_{G \setminus F}(u,v)$.
Then an $f$-EFT $t$-emulator must satisfy
\[\dist_{G \setminus F}(u,v) \leq \dist_{H^F}(u,v) \leq t \cdot \dist_{G \setminus F}(u,v)\]
for all $F \subseteq E$ with $|F| \leq f$ and $u,v \in V$.

Intuitively, this means a set of edge failures $F$ does not change the \emph{edge set} of $H$ at all; rather, its effect on $H$ is \emph{updated edge weights} due to increasing distances in $G \setminus F$.
Note that even when an edge $e$ is in both the emulator and the fault set, it remains as an edge in $H^F$; it just has a larger weight than it does in the pre-fault graph $H^{\emptyset}$.

This model may seem too strong: why shouldn't failures in $G$ affect the edges left in $H$, and why should the edges of $H$ automatically get reweighted to the shortest path in $G \setminus F$?

It turns out that this weight-updating behavior arises naturally in computer networking and distributed computing, particularly in overlay networks, which use virtual edges that inherit their weights from the routing topology of the underlying graph.
When pieces of the underlying graph fail, the virtual edges remain intact, but the underlying routing algorithm will automatically update to avoid failures and so the inherited lengths of the virtual edges will reconverge to the new shortest paths.
We refer the interested reader to~\cite{BDN22} for further practical and theoretical justification for this definition, and we also overview this discussion ourselves in Appendix~\ref{app:emulator-model}.

\subsection{Our Results}

\begin{table} \centering
\begin{tabular}{c|c|c}
& $(2k-1)$ Spanners & $(2k-1)$ Emulators\\
\hline
\multirow{4}{*}{Vertex Faults} & \multirow{4}{*}{$\Theta\left(f^{1-\frac{1}{k}}\right)$} & $\Theta\left(f^{\frac{1}{2}}\right)$ if $k=2$,\\
& & $\Thetaish_k\left( f^{\frac{1}{2} - \frac{1}{2k}}\right)$ if $k$ is odd,\\
& & $\Oish_k\left( f^{\frac{1}{2}} \right)$ and\\
& & $\Omega_k\left( f^{\frac{1}{2} - \frac{1}{2k}} \right)$ if $k \ge 4$ is even\\
\hline
\multirow{5}{*}{Edge Faults} & $\Theta\left(f^{\frac{1}{2}}\right)$ if $k=2$, & \color{blue} $\Theta\left(f^{\frac{1}{2}}\right)$ if $k=2$,\\
& $\Theta_k\left( f^{\frac{1}{2} - \frac{1}{2k}} \right)$ if $k$ is odd, & \color{blue} $ O\left(1 + f^{\frac{3}{5}} n^{-\frac{2}{15}}\right) $ if $k=3$,\\
& $O_k\left( f^{\frac{1}{2}} \right)$ and & \color{blue} $O_k\left(1 + f^{\frac{k}{k+2}} n^{- \frac{2}{k(k+2)}} + f^{\frac{1}{2}} n^{- \frac{1}{k(k+1)}}\right) $ if $k$ is odd\\
& $\Omega\left( f^{\frac{1}{2} - \frac{1}{2k}} \right)$ if $k \ge 4$ is even & \color{red} $\Theta(1)$ if $f = O\left(n^{\frac{2}{9}}\right), k=3$\\
& & \color{red} or $f = O\left( n^{\frac{2}{k(k+1)}} \right), k$ odd
\end{tabular}
\caption{\label{tbl:ftprice} Price of fault tolerance for edge and vertex fault tolerant spanners and emulators, proved in \cite{BDPW18,BP19,BDN22,BDR22}, including the new bounds proved in this paper (lower right).  All objects also have a lower bound of $\Omega(fn)$ edges, not reflected in this table.  We note that the upper bound of $O(f^{1/2})$ for EFT emulators when $k=2$ is \emph{not} proved in this paper. Rather, this is inherited from the other objects in this table; e.g., it follows from the definitions that any EFT spanner is also an EFT emulator.}
\end{table}

We show that EFT emulators allow for free fault-tolerance in a polynomial range of $f$.  
Our results, together with the previous work on fault-tolerant spanners and emulators, are summarized in Table~\ref{tbl:ftprice}.  
We will first state our main result in the special case $k=3$, for simplicity and because for a technical reason, our bound is slightly better here than for larger $k$.
\begin{theorem} [Main Result, $k=3$] \label{thm:3-upper-main}
Every $n$-node graph has an $f$-EFT $5$-emulator $H$ satisfying
$$|E(H)| = O\left(n^{4/3} + n^{6/5} f^{3/5} + nf\right).$$
\end{theorem}
\begin{corollary}
For any $n$-node graph and parameter $f = O(n^{2/9})$, there is an $f$-EFT $5$-emulator $H$ satisfying $|E(H)| = O(n^{4/3})$.
\end{corollary}
Recall that there are graphs that require $\Omega(n^{4/3})$ edges even for a non-fault-tolerant ($f=0$) $5$-emulator (see Theorem~\ref{thm:tightspanners}).
Thus in the parameter range $f = O(n^{2/9})$, the price of edge fault tolerance for $5$-emulators is $\Theta(1)$, i.e., \emph{free}.
We show a similar but slightly worse bound for general odd $k$:
\begin{theorem} [Main Result] \label{thm:k-upper-main}
For any fixed odd $k$, every $n$-node graph has an $f$-EFT $(2k-1)$-emulator $H$ satisfying
$$|E(H)| = O_k\left(n^{1+\frac{1}{k}} + n^{1+\frac{1}{k+2}} f^{\frac{k}{k+2}} + n^{1+\frac{1}{k+1}} f^{\frac{1}{2}} + fn\right).$$
\end{theorem}

Since our theorem assumes $k$ is fixed (a constant independent of $n$ and $f$) we use $O_k$ notation to hide factors that depend only on $k$.  We do not carefully track or try to optimize these factors, but it is not hard to see that they are at most polynomial in $k$.

This theorem implies the following range of free fault tolerance:
\begin{corollary}
For any fixed odd $k$, $n$-node graph, and parameter $f = O(n^{2/(k(k+1))})$, there is an $f$-EFT $(2k-1)$-emulator $H$ satisfying $|E(H)| = O_k(n^{1+1/k})$.
\end{corollary}
Recall that, assuming the girth conjecture, there are graphs which require $\Omega(n^{1+1/k})$ edges even for a non-fault tolerant $(2k-1)$-emulator.  Thus we can build emulators that are robust to $O(n^{2/(k(k+1))})$ faults for free.

The exact reason why our main result holds only for odd $k$ is rather technical, and it is discussed in more detail at the beginning of Section \ref{sec:limit}.
Very roughly, our strategy to upper bound emulator density is based on counting a certain kind of $k$-path in the emulator (see our proof overview in Section \ref{sec:overview}).
Our counting strategy produces a better bound when there exists an edge exactly in the middle of these $k$-paths, which is the case if and only if $k$ is odd.
This is a similar technical issue to the one that causes inherently different upper bounds for even/odd $k$ in the settings of edge fault tolerant spanners \cite{BDR22} and vertex fault tolerant emulators \cite{BDN22}; thus, this paper adds to an emerging theme that problems in this area seem harder to understand when $k$ is even.

We prove all our upper bounds constructively, via an algorithm which constructs EFT emulators of the claimed size.
While the ``main'' algorithm we analyze does not run in polynomial time (the obvious implementation of it would take time at least $\Omega(n^f)$), in Section \ref{sec:polytime} we show a variant that runs in polynomial time, paying an additional $O(k)$ factor in emulator size. This algorithm is based on ideas from~\cite{DR20}.

At a technical level, our size analysis requires a change in strategy from prior work.
The (near-optimal) size bounds in prior work for edge fault tolerant spanners and for vertex fault tolerant emulators happen to both be achieved by the method of analyzing \emph{alternating paths} \cite{BDR22, BDN22}; hence, this method has roughly been pushed to its limit.
In order to \emph{simultaneously} benefit from the focus on edge rather than vertex faults \emph{and} from the focus on emulators rather than spanners, we need a different approach entirely, which we overview in Section \ref{sec:overview}.

On the lower bounds side, we prove two limitations on the quality of EFT emulators.
Both bounds are based on constructions that have been repeatedly used in prior work \cite{ChechikLPR:10,BDPW18,BDR22}, but which need to be reanalyzed in the setting of EFT emulators.
\begin{theorem} [Tight Bounds for $k=2$] \label{thm:twoembounds}
Every $n$-node graph has an $f$-EFT $3$-emulator $H$ satisfying $|E(H)| = O(n^{3/2} f^{1/2})$.
This is best possible, in the sense that there exist graphs for which any $f$-EFT $3$-emulator $H$ has $|E(H)| = \Omega(n^{3/2} f^{1/2})$ 
\end{theorem}
The upper bound in Theorem \ref{thm:twoembounds} is inherited directly from prior work \cite{BDN22}, while the lower bound analysis is new in this paper.
A consequence of this theorem is that EFT $3$-emulators do \emph{not} confer free fault tolerance in any nontrivial range of parameters.
We also have:
\begin{theorem}
For any $k$, there are $n$-node graphs for which any $f$-EFT $(2k-1)$-emulator has $\Omega(fn)$ edges.
\end{theorem}

This lower bound is simple to prove and follows by considering any $f/2$-regular graph.
It implies that, for any $k$, the maximum conceivable range of free fault tolerance is $f = O(n^{1/k})$.

\subsection{Open Problems}

The main open question left by this work is to pin down the range of $f$ in which EFT emulators exhibit free fault tolerance.
There are two main sub-questions.
The first, and perhaps most pressing, is to determine whether free fault tolerance exists for even $k \ge 4$:
\begin{oq} [Is There Free Fault Tolerance for Even $k$?]
Is there any even constant $k \ge 4$ and constant $c > 0$ such that, for every $n$-node graph and parameter $f = O(n^c)$, there is an $f$-EFT $(2k-1)$-emulator $H$ satisfying $|E(H)| = O(n^{1+1/k})$?
\end{oq}
Our current arguments do not confer \emph{any} advantage for even $k$, relative to the size bounds inherited from EFT spanners and VFT emulators.
The technical barrier to improved bounds for even $k$ is discussed briefly in Section \ref{sec:limit}.
While it seems entirely plausible to apply our techniques more carefully and obtain improved edge bounds, we think that a significant new idea will be needed to prove (or refute) \emph{free} fault tolerance for any even $k$ in any polynomial range of $f$.

The second question is to determine the \emph{range} in which free fault tolerance exists for odd $k$:
\begin{oq} [What is the Range of Free Fault Tolerance for Odd $k$?]
For odd $k$, what is the largest constant $c_k$ such that, for every $n$-node graph and parameter $f = O(n^{c_k})$, there is an $f$-EFT $(2k-1)$-emulator $H$ satisfying $|E(H)| = O(n^{1+1/k})$?
\end{oq}
For this second question, this paper shows that
$$\Omega\left(\frac{1}{k^2}\right) \le c_k \le \frac{1}{k}.$$
It would be interesting to asymptotically improve the lower bound on $c_k$, improve the upper bound on $c_k$ at all, or settle the value of $c_3$ (we show that $2/9 \le c_3 \le 1/3$).
It would also be interesting to obtain ``nearly-free'' fault tolerance, e.g., subpolynomial price of fault tolerance in a polynomial range of $f$.

\section{Technical Overview and Outline} \label{sec:overview}
We begin with a high-level overview of our algorithm and proof strategy for constructing $f$-EFT $(2k-1)$-emulators.  

\subsection{Algorithm}
Algorithmically, we will do something simple that may nevertheless be surprising: we will run the obvious fault-tolerant variant of the greedy algorithm of Alth\"ofer et al.~\cite{AlthoferDDJS:93}.  That is, we will consider the edges of $G$ in nondecreasing weight order, and when considering an edge $(u, v)$ we will add it to our emulator if it is ``necessary''.  In the traditional non-fault tolerant setting, ``necessary'' means that $\dist_H(u,v) > (2k-1) \cdot w(u,v)$, i.e., the current spanner/emulator does not adequately preserve the distance between the endpoints $u, v$.  In the fault-tolerant spanner setting, where the greedy algorithm has also been used and analyzed~\cite{BDPW18,BP19,BDR22}, ``necessary'' means that there is some appropriate fault set $F$
such that $\dist_{H \setminus F}(u,v) > (2k-1) \cdot w(u,v)$ (i.e., the current solution does not adequately preserve the distance between the endpoints under fault set $F$).  Analogously, in our algorithm, ``necessary'' will mean that there is some fault set $F \subseteq E$ with $|F| \leq f$ (and not containing $(u,v)$) such that $\dist_{H^F}(u,v) > (2k-1) \cdot w(u,v)$.

Although this fault-tolerant greedy algorithm is standard in prior work, the reason it is surprising in our setting is that the emulator $H$ that we construct is actually a subgraph of the input graph $G$!
We note this does \emph{not} mean that $H$ is a fault tolerant spanner of $G$.

In particular, consider what happens if there is some edge $e=(u,v) \in E$ which we include in our emulator $H$, and then $e$ fails (is part of a fault set $F$).  In a fault-tolerant spanner, this would mean that $e$ also fails in the spanner: we would need $H \setminus F$ to be a spanner of $G \setminus F$, and so $e$ would be removed from both $G$ and $H$.
But in a fault-tolerant emulator, $e$ still exists in $H^F$: it just gets reweighted to $\dist_{G \setminus F}(u,v)$, and hence becomes more costly to use.
In other words, the relevant distance comparisons occur between $H^F$ and $G \setminus F$, and $H^F$ is not typically a subgraph of $G \setminus F$.
Thus even though our emulator $H$ consists only of edges that are also in $G$, we strongly use the fact that we are considering emulators rather than spanners.  

\subsection{Analysis}
The classical analysis of the non-fault tolerant greedy spanner algorithm works as follows.  

\begin{itemize}
\item First, one argues that $H$ has \emph{high girth}: in particular, all cycles in $H$ will have $>2k$ edges.
This is a fairly straightforward consequence of the greedy algorithm.

\item The focus of the proof then shifts from the particular output graph of the greedy algorithm to the more general class of graphs with high girth.
One applies a counting argument from extremal combinatorics called the \emph{Moore bounds}, stating that any $n$-node graph $H$ of girth $>2k$ can only have $O(n^{1+1/k})$ edges.
The Moore bounds are a counting argument over the \emph{non-backtracking $k$-paths} of $H$, and they are proved in two steps:
\begin{itemize}
\item (Counting Lemma) Letting $d$ be the average degree in $H$, and assuming $d$ is at least a large enough constant, we prove that $H$ has $n \cdot \Omega(d)^k$ non-backtracking $k$-paths.
\item (Dispersion Lemma) Say that a \emph{path meet} is a pair of distinct $k$-paths with the same endpoints.
We then prove that any $n$-node graph of girth $>2k$ cannot have any meets between two non-backtracking $k$-paths. 
\end{itemize}
The Moore bounds now follow roughly by arguing that the average degree $d$ must be only $O(n^{1/k})$, or else $H$ will have so many non-backtracking $k$-paths that (by the pigeonhole principle) two of them would share endpoints and thus form a meet.
\end{itemize}

In fault-tolerant spanners and emulators, it is unfortunately not true that the output graph $H$ has high girth.  However, one might intuitively believe that $H$ ``almost'' has high girth, and that this might be enough to prove some type of sparsity bound.  This proof strategy was first formalized by~\cite{BP19}, who introduced a notion of a ``blocking set'': informally, a collection of \emph{pairs} (either (edge, edge) pairs or (node,edge) pairs) that hits all cycles of length at most $2k$.  Intuitively, a graph that has a small blocking set ``almost'' has girth larger than $2k$, and so it is sometimes possible to prove similar sparsity bounds through appropriate generalizations of the counting lemma and the dispersion lemma.

These generalized lemmas often analyze more complicated paths: not just non-backtracking, but also satisfying a number of other technical conditions.  The precise type of paths that are considered is closely related to the precise version of blocking sets that are used.

We follow this high-level approach.  Our first technical innovation is a new kind of blocking set, which we call a \emph{double-blocking set}.  The precise definition is given later in Definition~\ref{def:doubleblocking}, but the idea is to require not just that cycles of length at most $2k$ are blocked, but also that cycles of length at most $k+1$ are blocked \emph{twice}.  We can show that the greedy emulator $H$ has a small double-blocking set (unlike fault-tolerant spanners).  This is slightly more difficult than the analogous part of previous work on spanners and VFT emulators, but is still the more straightforward part of our proof.

The goal is then to show that any graph with this stronger kind of blocking set is particularly sparse.  Following the outline for the Moore bounds, we need to define some type of path and prove a counting lemma and dispersion lemma.  We introduce a new type of path, which we call SCUM paths (Simple, Chain Unblocked, Middle-Heavy), and prove the associated counting and dispersion lemmas (see Definition~\ref{def:scum} for the precise definition).  The chain unblocked property is our main new technical idea, as it differs substantially from the technical conditions on paths used in prior work.  Once we have the definition of these paths, the counting lemma is reasonably straightforward, but the dispersion lemma is more technically involved.

\subsection{Outline}
We begin in Section~\ref{sec:algorithm} by formally stating our algorithm, proving that the graph $H$ it returns is an $f$-EFT $(2k-1)$-emulator, and then proving that $H$ has a double-blocking set of size at most $f|E(H)|$.  Then in Section~\ref{sec:analysis-main}, we begin the proof that \emph{any} graph with a small double-blocking set is sparse enough to satisfy our main theorem statement.  We do this by proving a counting lemma for SCUM paths in Section~\ref{sec:counting}, stating but not proving dispersion lemmas for $k=3$ and for general odd $k$ in Section~\ref{sec:dispersion}, and then showing in Section~\ref{sec:wrapup} how to combine the counting and dispersion lemmas to prove that $H$ is sparse.
We then prove our dispersion lemmas for $k=3$ in Section~\ref{sec:limit3}, and for general odd $k$ in Section~\ref{sec:limit}.  In Section~\ref{sec:polytime} we show how to modify the algorithm to run in polynomial time without significantly affecting the sparsity bounds.  Finally, in Section~\ref{sec:lower} we construct our lower bounds on the size of EFT emulators.

\section{Algorithm and Double-Blocking} \label{sec:algorithm}

We begin by formally giving our algorithm, proving that it does output an EFT emulator, and then proving that the EFT emulator it outputs has a small double-blocking set.

\subsection{Algorithm and Correctness}

We construct our emulators using the following greedy algorithm.  Recall that $H^F$ is the version of the graph $H$ with weights updated according to distances in $G \setminus F$.  So $H^F$ is defined with respect to an ambient graph $G$, which will typically be clear from context.

\FloatBarrier
\begin{algorithm} [ht]
\textbf{Input:} Graph $G = (V, E, w)$, positive integers $f, k$\;

Let $H \gets (V, \emptyset)$ be the initially-empty emulator\;
\ForEach{edge $(u, v) \in E$ in order of nondecreasing weight $w(u, v)$}{
    \If{there is $F \subseteq E$ of size $|F| \le f$ with $\dist_{H^F}(u, v) > (2k-1) \cdot w(u, v)$}{
        Add $(u, v)$ to $H$\;
    }
}
\Return{$H$};
\caption{Algorithm for $f$-EFT $(2k-1)$-emulators}
\label{alg:greedy}
\end{algorithm}
\FloatBarrier

If two edges in the input graph have the same weight, then the algorithm considers them in arbitrary order.  

Correctness of the output emulator is a standard argument:
\begin{theorem} [Correctness of Algorithm \ref{alg:greedy}]
The graph $H$ returned by Algorithm \ref{alg:greedy} is an $f$-EFT $(2k-1)$ emulator of the input graph $G$.
\end{theorem}
\begin{proof}
Let $F$ be an arbitrary set of $|F| \le f$ failing edges.
Our goal is to prove that $H^F$ is a $(2k-1)$-emulator of $G \setminus F$.
By a standard observation, it suffices to prove that
$\dist_{H^F}(u, v) \le (2k-1) \cdot w(u, v)$
for all edges $(u, v) \in G \setminus F$ (by considering shortest paths, this implies the desired stretch bound for all pairs).
There are two cases.  First, if $(u, v) \in H$, then we immediately have $\dist_{H^F}(u, v) = \dist_{G\setminus F}(u, v) \le w(u, v)$, due to the reweighting of $(u, v)$ in $H^F$.  Second, suppose $(u, v) \notin H$.
We must have considered the edge $(u, v)$ at some point in the algorithm, and then elected not to add it to $H$.
By the algorithm conditional, this implies that for every $F$ we have $\dist_{H^F}(u, v) \le (2k-1) \cdot w(u, v)$, completing the proof. 
\end{proof}

\subsection{Double-Blocking Sets}

In the following analysis, it will frequently be important to consider a subset of edges within $H$ (typically a cycle or path) and consider the last among these edges to be processed by the main loop of our algorithm.
We will use the phrase ``the latest (or last) edge (in the ordering of $H$)'' to identify this edge.  Note that if all weights are distinct then this will be the highest-weight edge in the subset, but if some edges have the same weight then it specifically refers to the highest-weight edge under the same tiebreaking used by the algorithm.

We begin with our new definition of double-blocking sets.

\begin{definition} [Double-Blocking Sets (see Figure \ref{fig:doubleblock} for $k=3$)] \label{def:doubleblocking}
Let $H = (V, E)$ be a graph with a total ordering on its edges, let $k$ be a positive integer, and let $B \subseteq \binom{E}{2}$.
We say that a cycle $C$ in $H$ is \emph{blocked by} $\{e_1, e_2\} \in B$ if we have $e_1, e_2 \in C$ and $e_1$ is the last edge in $C$ in the ordering.
The set $B$ is called a \emph{$2k$ double-blocking set} for $G$ if:
\begin{itemize}
\item Every cycle $C$ on $|C| \le 2k$ edges is blocked by some $b \in B$, and

\item Every cycle $C$ on $|C| \le k+1$ edges is blocked by two distinct $b_1, b_2 \in B$.
\end{itemize}
The edge pairs in a blocking set are called \emph{blocks}.
\end{definition}

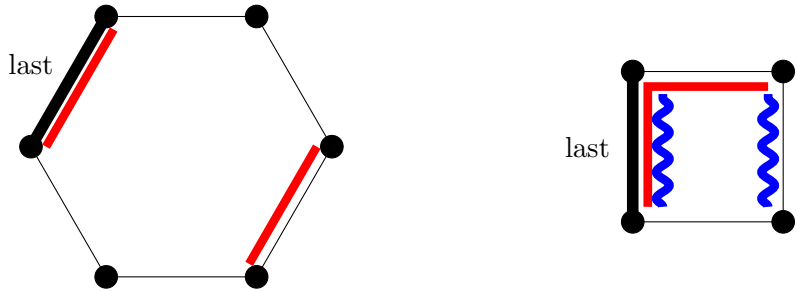
\begin{figure} [ht] \centering
\begin{tikzpicture}
\draw [fill=black] (2, 0) circle [radius=0.15];
\draw [fill=black] (1, 1.732) circle [radius=0.15];
\draw [fill=black] (1, -1.732) circle [radius=0.15];
\draw [fill=black] (-2, 0) circle [radius=0.15];
\draw [fill=black] (-1, 1.732) circle [radius=0.15];
\draw [fill=black] (-1, -1.732) circle [radius=0.15];

\draw (-2, 0) -- (-1, 1.732) -- (1, 1.732) -- (2, 0) -- (1, -1.732) -- (-1, -1.732) -- cycle;

\node at (-2.0, 1.1) {last};
\draw [line width = 0.4em] (-2, 0) -- (-1, 1.732);

\draw [red, line width = 0.3em] (-1.8, 0) -- (-0.9, 1.559);
\draw [red, line width = 0.3em] (1.8, 0) -- (0.9, -1.559);

\begin{scope}[shift={(7, 0)}]
\draw [fill=black] (1, 1) circle [radius=0.15];
\draw [fill=black] (-1, -1) circle [radius=0.15];
\draw [fill=black] (1, -1) circle [radius=0.15];
\draw [fill=black] (-1, 1) circle [radius=0.15];

\draw (1, 1) -- (1, -1) -- (-1, -1) -- (-1, 1) -- cycle;

\draw [line width = 0.4em] (-1, 1) -- (-1, -1);
\node at (-1.6, 0) {last};
\draw [red, line width = 0.3em] (-0.8, -0.8) -- (-0.8, 0.8) -- (0.8, 0.8);
\draw [blue, snake it, line width = 0.3em] (-0.6, -0.8) -- (-0.6, 0.7);
\draw [blue, snake it, line width = 0.3em] (0.8, -0.8) -- (0.8, 0.7);

\end{scope}
\end{tikzpicture}
\caption{\label{fig:doubleblock} (Left) For any cycle $C$ with $\le 6$ edges, a $6$ double-blocking set $B$ blocks $C$ at least once, e.g., by the pair of red edges.  (Right) Additionally, for any cycle $C$ with $\le 4$ edges, $B$ blocks $C$ at least twice, e.g., by the pair of red edges and also the pair of wavy blue edges.}
\end{figure}

We note that the ``short'' cycles $C$ of length $k+1$ have blocks of the form $\{e_1, e_2\}, \{e_1, e_3\}$; that is, both blocks must contain the last edge in $C$ and thus they intersect on an edge.  We now want to show that the output emulator $H$ of the greedy algorithm has a small double-blocking set.

\begin{lemma} \label{lem:doubleblocking}
The emulator $H$ returned by Algorithm \ref{alg:greedy} has a $2k$ double-blocking set $B$ of size $|B| \le f |E(H)|$ (with respect to the order that the edges are added to $H$ in Algorithm \ref{alg:greedy}).
\end{lemma}
\begin{proof}
Each time we add an edge $(u, v)$ to $H$, we do so because of an edge subset $F_{(u, v)}$ with $\dist_{H^{F_{(u, v)}}}(u, v) > (2k-1) \cdot w(u, v)$.  Note that $F_{(u,v)}$ contains only edges earlier in the ordering that $(u,v)$, and hence in particular does not contain $(u,v)$. 
Define
$$B := \left\{ \{e, f\} \ \mid \ e \in E(H), f \in F_e \right\}.$$
Since $|F_e| \le f$ for all $e$, we have $|B| \le f|E(H)|$.
We now show that $B$ is a $2k$ double-blocking set.
\begin{itemize}
\item Let $C$ be a cycle in $H$ of length $\le 2k$, and let $(u, v) \in C$ be its last edge.
When we add $(u, v)$ to $H$, the other $2k-1$ edges in $C$ form a $u \leadsto v$ path of length $\le (2k-1) \cdot w(u, v)$.
Thus, we must include one of these edges $f \in C \setminus \{(u, v)\}$ in $F_{(u, v)}$, to fault this short $u \leadsto v$ path.
We thus have $\{(u, v), f\} \in B$, which blocks $C$ as required.

\item (See Figure \ref{fig:doublecase}) Now let $C$ be a cycle in $H$ of length $\le k+1$, and let $(u, v) \in C$ be its last edge.
By the previous item, there exists at least one edge $(x, y) \in F_{(u, v)} \cap C$.
Suppose towards contradiction that $(x, y)$ is the \emph{only} edge in $F_{(u, v)} \cap C$.
Then:
\begin{itemize}
\item In $G \setminus F_{(u, v)}$, we have $\dist(x, y) \le k \cdot w(u, v)$.
This is due to the $x \leadsto y$ path that uses the other $\le k$ edges in $C$, each of which has weight $\le w(u, v)$.  Note that this crucially uses the ``automatic edge updating'' in our model.

\item We have
$$\dist_{H^{F_{(u, v)}}}(u, v) \le (2k-1) \cdot w(u, v).$$
This is due to the $u \leadsto v$ path that uses the other $\le k$ edges in $C$.
The $\le k-1$ edges besides $(x, y)$ each have weight $\le w(u, v)$, and by the previous item, the edge $(x, y)$ has weight $\le k \cdot w(u, v)$ in $H^{F_{(u,v)}}$. 
\end{itemize}
The latter point implies that the algorithm conditional is violated and the fault set $F_{(u, v)}$ would \emph{not} cause us to add $(u, v)$ to $H$, completing the contradiction.  Hence there are at least two distinct edges $f_1, f_2 \in F_{(u,v)} \cap C$.  Since $(u,v) \not \in F_{(u,v)}$, this means that both $\{(u,v), f_1\}$ and $\{(u,v), f_2\}$ are in $B$ by definition, so $C$ is double-blocked as required. \qedhere
\end{itemize}
\end{proof}

\begin{figure} [ht] \centering
\begin{tikzpicture}
\draw [fill=black] (1, 1) circle [radius=0.15];
\draw [fill=black] (-1, -1) circle [radius=0.15];
\draw [fill=black] (1, -1) circle [radius=0.15];
\draw [fill=black] (-1, 1) circle [radius=0.15];

\draw (1, 1) -- (1, -1) -- (-1, -1) -- (-1, 1) -- cycle;

\draw [line width = 0.4em] (-1, 1) -- (-1, -1);

\node at (-1.4, 1) {$u$};
\node at (-1.4, -1) {$v$};
\node at (1.4, 1) {$x$};
\node at (1.4, -1) {$y$};
\node [red] at (1, 0) {\bf \huge $\times$};
\draw [red, ultra thick, ->] (1, 0.9) arc (45:315:1.2);
\node [red] at (2.5, 0) {$\le 3 \cdot w(u, v)$};

\begin{scope}[shift={(10, 0)}]
\draw [fill=black] (1, 1) circle [radius=0.15];
\draw [fill=black] (-1, -1) circle [radius=0.15];
\draw [fill=black] (1, -1) circle [radius=0.15];
\draw [fill=black] (-1, 1) circle [radius=0.15];

\draw (1, 1) -- (1, -1) -- (-1, -1) -- (-1, 1) -- cycle;

\draw [line width = 0.4em] (-1, 1) -- (-1, -1);

\node at (-1.4, 1) {$u$};
\node at (-1.4, -1) {$v$};
\node at (1.4, 1) {$x$};
\node at (1.4, -1) {$y$};
\node [red] at (1, 0) {\bf \huge $\times$};
\draw [blue, ultra thick, ->] (-1, 0.9) arc (135:-135:1.2);
\node [blue] at (-2.5, 0) {$\le 5 \cdot w(u, v)$};

\node [red] at (2.5, 0) {$\le 3 \cdot w(u, v)$};

\end{scope}
\end{tikzpicture}
\caption{ \label{fig:doublecase} In the proof of Lemma \ref{lem:doubleblocking}, for the case $k=3$, we argue double-blocking of $4$ cycles as follows.  Suppose our greedy algorithm is currently considering whether to add $(u, v)$, the heaviest edge in the cycle, to the emulator.  (Left) If there is only a single edge $(x, y)$ on the cycle included in the set of failing edges, then $w(x, y)$ increases to at most $3 \cdot w(u, v)$.  (Right) This implies $\dist(u, v) \le 5 \cdot w(u, v)$, so we would not choose to add $(u, v)$ to the emulator in the first place.}
\end{figure}
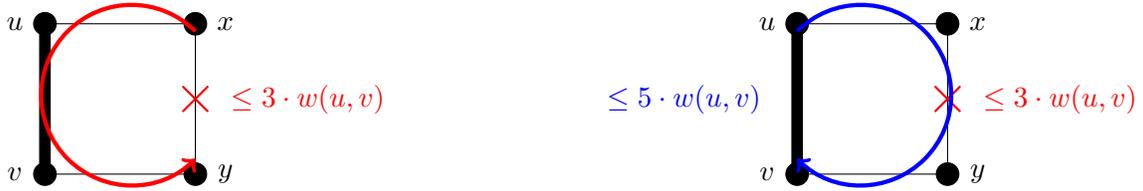

\section{Main Size Analysis (assuming dispersion)} \label{sec:analysis-main} 
We now shift perspective.
Instead of talking about the greedy algorithm or fault tolerance, our goal in the rest of the proof is simply to limit the maximum number of edges that can appear in an $n$-node graph $H$ with a $2k$ double-blocking set of size $|B| \le f|E(H)|$.  In particular, we will prove the following theorems:

\begin{theorem}[Sparsity for $k=3$] \label{thm:3-blocked-upper}
Let $H$ be an $n$-node graph with an ordering on its edges which has a $6$ double-blocking set of size at most $f|E(H)|$.  Then $|E(H)| = O(n^{4/3} + n^{6/5} f^{3/5} +fn)$.
\end{theorem}

\begin{theorem}[Sparsity for general odd $k$] \label{thm:k-blocked-upper}
Let $k$ be a positive odd integer and let $H$ be an $n$-node graph with an ordering on its edges which has a $2k$ double-blocking set of size at most $f|E(H)|$. Then
\[|E(H)| = O_k\left(n^{1+\frac{1}{k}} + n^{1+\frac{1}{k+2}} f^{\frac{k}{k+2}} + n^{1+\frac{1}{k+1}} f^{\frac{1}{2}} + fn\right).\]
\end{theorem}

Theorem~\ref{thm:3-blocked-upper} and Lemma~\ref{lem:doubleblocking} imply Theorem~\ref{thm:3-upper-main} (our main bound for $k=3$), and Theorem~\ref{thm:k-blocked-upper} and Lemma~\ref{lem:doubleblocking} imply Theorem~\ref{thm:k-upper-main} (our main bound for general odd $k$).  So it just remains to prove these two theorems.

\subsection{Graph Cleaning} \label{sec:cleaning}
We first use some standard combinatorial arguments to make some additional assumptions on $H$ without loss of generality.  In particular, since our goal is to prove an upper bound on $|E(H)|$ based on the fact that it has $2k$ double-blocking set of size $|B| \le f|E(H)|$, we may assume without loss of generality that the following are true:
\begin{itemize}
    \item Every edge in $H$ participates in at most $f$ pairs in $B$.
    \item If $d$ is the average degree of a node in $H$, then every node has degree $\Theta(d)$.
    \item $f \leq d/(Ck)$ for some large constant $C$. Note that we are assuming $k$ is a constant, so this is equivalent to saying that $f \leq d/C'$ for a large constant $C'$.  
\end{itemize}

The proof that we can assume these without loss of generality is relatively standard, so we defer it to Appendix~\ref{app:cleaning}.  From now on, we will make these assumptions on $H$.
We also carry forward the notation $d$ for the average degree in $H$.

\subsection{SCUM Path Counting} \label{sec:counting}

As discussed in Section~\ref{sec:overview}, we introduce a new type of path to analyze, called SCUM paths. 

\begin{definition}[SCUM $k$-paths] \label{def:scum}
Let $H$ be a graph with an ordering of its edges, and let $B \subseteq \binom{E(H)}{2}$ be a $2k$ double-blocking set.
A SCUM $k$-path is a path $\pi$ with $k$ edges that satisfies all of the following properties.
Let $(e_1 = (v_0, v_1), \dots, e_k = (v_{k-1}, v_k))$ be the edge sequence of $\pi$, and let $m := (k+1)/2$,\footnote{We will focus on odd $k$, where this is integral.  For even $k$, one would round up or down arbitrarily.} such that $e_m$ is the middle edge of $\pi$.
\begin{itemize}
\item (Simple) $\pi$ does not repeat nodes,
\item (Chain Unblocked) Let us say that an oriented path $q$ is \emph{chain unblocked} (with respect to $B$) if, for any edge $e \in q$ and node $v \in q$ where $v$ strictly follows $e$ (in particular $v \notin e$), there is no block of the form $\{e, e'\} \in B$ where $e'$ is incident on $v$.
We require:
\begin{itemize}
\item The oriented subpath of $\pi$ with edge sequence $(e_m, e_{m-1}, \dots, e_1)$ is chain-unblocked, and
\item The oriented subpath of $\pi$ with edge sequence $(e_m, e_{m+1}, \dots, e_k)$ is chain-unblocked.
\end{itemize}
\item (Middle-Heavy) the latest edge of $\pi$ (in the edge ordering of $H$) is $e_m$.

\end{itemize}
\end{definition}

The main new ingredient in this definition is the chain unblocked property.
Prior work has often focused on \emph{unblocked paths} instead, e.g., where there is no $\{e, e'\} \in B$ with $e, e' \in \pi$ \cite{BDR22, BDN22}.

The chain-unblocked property is qualitatively different and not directly comparable: it is possible for a path to be chain-unblocked and yet blocked in the traditional sense, and it is also possible for a path to be chain-blocked and yet unblocked in the traditional sense.

We now want to prove a counting lemma for our SCUM $k$-paths.
The first technical step is to classify the edges $e \in E(H)$ into \emph{core edges} and \emph{non-core edges}.
At the moment an edge $e \in E(H)$ is added to $H$, we consider the induced subgraph $H_e \subseteq H$ obtained by iteratively deleting any remaining nodes of degree $\le d/4$.
If $e$ survives in $H_e$, then it is a \emph{core edge}, otherwise it is a \emph{non-core edge}.
We have:
\begin{claim} \label{clm:halfcore}
At least half the edges of $H$ are core edges.
\end{claim}
\begin{proof}
Let $H_{nc}$ be the subgraph that contains exactly the non-core edges of $H$.
Notice that, by definition of non-core edges, if we iteratively delete nodes in $H_{nc}$ of degree $\le d/4$ then every edge will eventually be deleted from $H_{nc}$.
By unioning over the $n$ nodes in $H_{nc}$, we delete at most $n \cdot d/4 \le |E(H)|/2$ edges in this process.
Thus, $\le |E(H)|/2$ of the edges in $H$ are non-core edges, and the claim follows.
\end{proof}

\begin{lemma} [SCUM $k$-path Counting Lemma] \label{lem:scumcount}
There are $n \cdot \Omega(d)^k$ SCUM $k$-paths in $H$.
\end{lemma}
\begin{proof}
We may generate a SCUM $k$-path $\pi$ as follows:
\begin{itemize}
\item Choose a core edge $e_m = (u, v)$ to act as the middle edge of $\pi$.  Since $\pi$ will be a $k$-path and $k$ is odd, this means that $m = (k+1)/2$.
Using Claim \ref{clm:halfcore}, there are $\Omega(nd)$ choices for $e_m$.
Let $H_{e_m} \subseteq H$ be the subgraph witnessing that $e_m$ is a core edge: that is, consider the subgraph of $H$ that contains only $e_m$ and preceding edges in the ordering, and then $H_{e_m}$ is the induced subgraph obtained by iteratively deleting nodes of degree $\le d/4$.

\item Sequentially choose $ m-1 $ edges $e_{m-1}, e_{m-2}, \dots, e_1 \in E(H_{e_m})$ that extend $e_m$ into a path of length $m$.
If $e_i = (x, y)$, the next edge $e_{i-1}$ may be selected among the $\Omega(d)$ edges incident to $y$ in $H_{e_m}$, disallowing (1) the $\le k$ edges that connect $y$ to a node that already appears in $\pi$ (to enforce simplicity), and (2) the $\le fk$ edges that connect $y$ to nodes $z$, such that an edge already in $\pi$ is blocked with an edge incident to $z$ (to enforce that $\pi$ is chain unblocked). 
We have $fk \leq d/C$ for a large constant $C$, since we assume $f \leq d/(Ck)$ (see Section~\ref{sec:cleaning}). 
So there are $\Omega(d)$ choices for the next edge at each step, for a total of $\Omega(d)^{m-1}$ choices.

\item Repeat the above process to sequentially choose edges in the suffix $e_{m+1}, \dots, e_k$.
By identical logic, there are $\Omega(d)$ choices at each step, for a total of $\Omega(d)^{m-1}$ choices.
\end{itemize}
The SCUM properties follow from the construction.
Since $e_m$ is the middle edge of a path of length $k$, and $k$ is odd, we have $m = (k+1)/2$.
Thus we count
$$\Omega(nd) \cdot \Omega(d)^{m-1} \cdot \Omega(d)^{m-1} = n \cdot \Omega(d)^k$$
ways to generate a SCUM $k$-path by this process.
\end{proof}

\subsection{Limiting SCUM Path Meets (Dispersion)} \label{sec:dispersion}
Now that we have a \emph{lower} bound on the number of SCUM paths, we want to find an \emph{upper} bound on a related but distinct object: SCUM $k$-path \emph{meets}.  Combining the two will lead to the desired sparsity bounds.  

\begin{definition} [SCUM Path Meets]
A SCUM $k$-path meet is a pair of distinct SCUM $k$-paths $\{\pi_1, \pi_2\}$ that use the same endpoints $s, t$.
\end{definition}

Proving an upper bound on SCUM $k$-path meets is the most technically complex part of our argument, so we defer the proofs to Section~\ref{sec:limit3} for $k=3$ and Section~\ref{sec:limit} for general odd $k$.  We separate them since the bound for $k=3$ is both stronger and easier to prove than the bound for general $k$.  For now, we simply state the lemmas, and then in Section~\ref{sec:wrapup} we will show how to use them to prove our main theorems.

\begin{lemma}[Dispersion lemma for $k=3$] \label{lem:3meetbound}
When $k=3$, there are $O( nd f^3 )$ SCUM $3$-path meets in $H$.
\end{lemma}

\begin{lemma}[Dispersion lemma for odd $k$] \label{lem:oddmeetbound}
When $k$ is odd, the number of SCUM $k$-path meets in $H$ is
$$O_k\left(n f^{k} d^{k-2} + n f^{(k+1)/2} d^{k-1} \right).$$
\end{lemma}

\subsection{Proof Wrapup} \label{sec:wrapup}

The rest of the proof is a matter of comparing our lower bound on the number of SCUM $k$-paths to the upper bound on the number of SCUM $k$-path meets, and algebraically rearranging our expressions into a bound on $d$.
We return to general $k$ for the following lemma.
\begin{lemma} \label{lem:csmeets}
For any parameter $\alpha \ge 0$, if $H$ has $\le \alpha n^2$ SCUM $k$-path meets, then $d \le O\left(n^{1/k} + n^{1/k} \alpha^{1/(2k)} \right)$.
\end{lemma}
\begin{proof}
Let $x_{s, t}$ be the number of $s \leadsto t$ SCUM $k$-paths.
We have:
\begin{align*}
\sum \limits_{(s, t)} x_{(s, t)} &= \left|\left\{ (s, t) \ \mid \ x_{(s, t)} = 1 \right\} \right| + \left( \sum \limits_{(s, t) | x_{(s, t)} \ge 2} x_{(s, t)} \right)\\
&\le n^2 + O\left( n^2 \cdot \sum \limits_{(s, t)} \binom{x_{(s, t)}}{2} \right)^{1/2} \tag*{Cauchy-Schwarz}\\
&\le n^2 + O\left( n^2 \cdot \alpha n^2 \right)^{1/2} \\
&= O\left(n^2 + n^2 \alpha^{1/2}\right).
\end{align*}
We note that the second-to-last inequality is because, by definition, the number of $s \leadsto t$ SCUM $k$-path meets is exactly $\binom{x_{(s, t)}}{2}$.
Thus, summing $\binom{x_{(s, t)}}{2}$ over all pairs $(s, t)$ counts the number of SCUM $k$-path meets, which is at most $\alpha n^2$.

By Lemma \ref{lem:scumcount}, we also have the lower bound $n \cdot \Omega(d)^k$ on the total number of SCUM $k$-paths.
Thus,
$$n \cdot \Omega(d)^k \le O\left(n^2 + n^2 \alpha^{1/2} \right).$$
If the former term on the right-hand side dominates, then we rearrange to get $d = O(n^{1/k})$.
If the latter term on the right-hand side dominates, then we rearrange to get $d = O(n^{1/k} \alpha^{1/(2k)})$.
\end{proof}

\subsubsection{\texorpdfstring{$k=3$}{k=3}}
We now again restrict to $k=3$, and we wrap up the proof of Theorem \ref{thm:3-blocked-upper}. 
Plugging in the bound in Lemma~\ref{lem:3meetbound} to the previous lemma, we have
$$d \le O\left( n^{1/3} +  n^{1/3} \left(df^3 / n \right)^{1/6} \right).$$
If the former term dominates, then we get $d = O(n^{1/3})$.
If the latter term dominates, then we continue
$$d^{5/6} \le O\left( n^{1/6} f^{1/2} \right)$$
and so
$$d \le O\left( n^{1/5} f^{3/5} \right)$$
which completes the proof of Theorem~\ref{thm:3-blocked-upper}.

\subsubsection{General Odd \texorpdfstring{$k$}{k}}

We now complete the proof of Theorem~\ref{thm:k-blocked-upper} for general $k$.
By Lemma \ref{lem:oddmeetbound}, the number of SCUM $k$-path meets is
$$O_k\left(n f^{k} d^{k-2} + n f^{(k+1)/2} d^{k-1} \right).$$
Suppose the first term dominates.
Plugging into Lemma \ref{lem:csmeets}, we get
$$d = O(n^{1/k}) + n^{1/k} \cdot O_k\left(\frac{f^{k} d^{k-2}}{n}\right)^{\frac{1}{2k}}.$$
If the former term dominates, we get $d = O(n^{1/k})$. 
If the latter term dominates, we continue
$$d^{k+2} = n \cdot O_k\left(f^{k}\right)$$
and so
$$d = O_k\left(f^{\frac{k}{k+2}} n^{\frac{1}{k+2}}\right).$$
On the other hand, returning to our bound on the number of SCUM $k$-path meets, suppose the second term dominates.
Again plugging into Lemma \ref{lem:csmeets}, we get
$$d = O(n^{1/k}) + n^{1/k} \cdot O_k\left(\frac{f^{(k+1)/2} d^{k-1}}{n}\right)^{\frac{1}{2k}}.$$
If the first term dominates we again get $d = O(n^{1/k})$.
If the latter term dominates, we continue
$$d^{k+1} = O_k\left(n \cdot f^{(k+1)/2}\right)$$
and so
$$d = O_k\left(f^{1/2} n^{\frac{1}{k+1}} \right).$$

\section{Limiting SCUM Path Meets (\texorpdfstring{$k=3$}{k=3})} \label{sec:limit3}
In this section we prove Lemma~\ref{lem:3meetbound}, which as discussed implies Theorem~\ref{thm:3-upper-main}.

\subsection{Technical Lemmas on Path Structure}

Our focus in this part is on $k=3$ specifically (i.e., graphs with a $6$ double-blocking set).
In this setting, the necessary technical lemmas about the structure of our SCUM paths become simpler to state and use.
The most important simplification is reflected by Lemma \ref{lem:3scumdisjoint}, which shows that in any SCUM $3$-path meet, the two paths of the meet are internally vertex-disjoint.
This lemma does \emph{not} generalize to larger $k$, which introduces technical complexities and makes our quantitative bounds a bit worse.
Analogous structural lemmas for general $k$ are given in Section \ref{sec:oddtech}.

Our goal is to prepare for a dispersion lemma limiting the number of SCUM $3$-path meets.
We note that, when considering a graph with a $2k$ blocking set, one generally focuses on meets between $k$-paths (not $2k$-paths).
This is analogous to the Moore bounds, and the reason is roughly because two $k$-paths in a meet create a cycle on $\le 2k$ edges, and so our blocking set is indeed helpful in analyzing the structure of this meet.

The following three lemmas hold only for $k \ge 3$ because they rely on the set of blocks in $B$ (either explicitly or by discussion of SCUM paths), and we will use that $B$ is a $(\ge 6)$ double blocking set in their proofs.

\begin{lemma} \label{lem:3midedge}
Suppose $k \ge 3$, let $\{\pi^\ell, \pi^r\}$ be a SCUM $3$-path meet, let $e^*$ be the middle edge of $\pi^\ell$, and suppose $e^*$ is the latest edge in the ordering in $\pi^\ell \cup \pi^r$.
Let $F^*$ be the set of edges that appear in a block with $e^*$ in $B$.
Then the middle edge of $\pi^r$ is in $F^*$.
\end{lemma}
\begin{proof}
Notice that $e^* \notin \pi^r$, since this would imply $\pi^\ell=\pi^r$ but $\pi^\ell, \pi^r$ must be distinct in a meet.
Thus, the paths $\pi^\ell \cup \pi^r$ form a cycle $C$ of length $\le 6$, and this cycle contains $e^*$ as its heaviest edge.
We therefore have a block of the form $\{e^*, e\} \in B$ with $e \in C$.
Since $\pi^\ell$ is chain unblocked, $e$ cannot be incident on the first or last node of $\pi^\ell$.
Thus, the only possible placement of $e$ in $\pi^\ell \cup \pi^r$ is as the middle edge of $\pi^r$.
\end{proof} 

\begin{lemma} \label{lem:3scumdisjoint}
Suppose $k \ge 3$.
In any SCUM $3$-path meet $\{\pi^\ell, \pi^r\}$, the paths $\pi^\ell, \pi^r$ are internally vertex disjoint.
(That is, they intersect only on their endpoints, and otherwise do not share nodes.)
\end{lemma}
\begin{proof}
Let $s, t$ be the endpoints of $\pi^\ell, \pi^r$, let $e^* = (u^*, v^*)$ be the latest edge in the ordering in $\pi^\ell \cup \pi^r$, and suppose without loss of generality that $e^*$ is the middle edge of $\pi^\ell$.
Notice that $e^*$ cannot be the middle edge $e'$ of $\pi^r$, since this would imply $\pi^\ell=\pi^r$.

Seeking contradiction, suppose that $\pi^\ell, \pi^r$ intersect on their first edge $(s, u^*)$.
Then the last two edges of $\pi^\ell, \pi^r$ form a $4$-cycle $C$ with nodes $(u^*, v^*, t, x)$, which must be double-blocked in $B$.
It is possible that the middle edges $\{e^* = (u^*, v^*), (u^*, x)\}$ appear together in $B$, which would block $C$ once.
But, since $\pi^\ell$ is chain unblocked, $e^*$ cannot appear in $B$ with either of the other edges $\{(v^*, t), (t, x)\}$ which completes the contradiction.

The argument that $\pi^\ell, \pi^r$ cannot intersect on their last edges $(v^*, t)$ is symmetric.
\end{proof}

\begin{lemma} \label{lem:2pathlim}
When $k \ge 3$, for any nodes $s, t$, there are at most $f+1$ simple $s \leadsto t$ paths in $H$ of length $2$.
\end{lemma}
\begin{proof}
Let $\Pi_{s, t}$ be the set of $s \leadsto t$ paths of length $2$.
Notice that any two paths in $\Pi_{s, t}$ form a $4$-cycle, which must be double-blocked.
Let $e^*$ be the latest edge in the ordering contained in any path in $\Pi_{s, t}$.
Then for any path $\pi \in \Pi_{s, t}$ with $e^* \notin \pi$, there exists a block of the form $\{e^*, e\}$ with $e \in \pi$.
Since we may assume without loss of generality that $e^*$ is only blocked with at most $f$ other edges (see Section~\ref{sec:cleaning}), it follows that $|\Pi_{s, t}| \le f+1$.
\end{proof}

\subsection{Proof of Lemma~\ref{lem:3meetbound}}
We can generate all SCUM $3$-path meets $\{\pi^\ell, \pi^r\}$ by the following process (see Figure \ref{fig:3meetgen}).
\begin{itemize}
\item Choose an edge $e^* = (u^*, v^*)$ to act as the latest edge in $\pi^{\ell} \cup \pi^r$ in the ordering.
We will assume without loss of generality that $e^*$ is the middle edge of $\pi^\ell$.
There are $O(nd)$ possible choices of $e^*$.

\item Let $F^*$ be the edges that are paired with $e^*$ in $B$.
By Lemma \ref{lem:3midedge}, $\pi^r$ uses an edge $e' = (x, y) \in F^*$ as its middle edge.
There are $O(f)$ possible choices of $e'$.

\item By Lemma \ref{lem:3scumdisjoint}, $\pi^\ell, \pi^r$ are internally vertex disjoint, so $u^* \ne x$ and $v^* \ne y$.
Thus, we may complete the meet by choosing a simple $u^* \leadsto x$ $2$-path and a simple $v^* \leadsto y$ $2$-path.
By Lemma \ref{lem:2pathlim}, there are $O(f)$ ways to choose each $2$-path, and thus $O(f^2)$ ways to complete the meet in total.
\end{itemize}

Putting it together, we have $O(nd \cdot f \cdot f^2)$ meets.

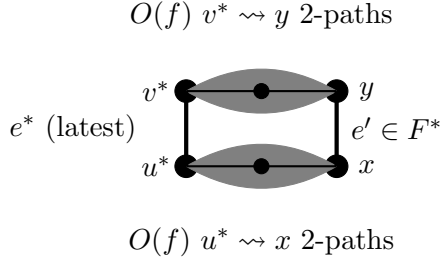
\begin{figure} [h] \centering
\begin{tikzpicture}
\draw [fill=black] (0, 0) circle [radius=0.15];
\draw [fill=black] (0, 1) circle [radius=0.15];

\draw [fill=black] (2, 0) circle [radius=0.15];
\draw [fill=black] (2, 1) circle [radius=0.15];

\draw [gray, fill=gray] (0, 1) to[bend left=30] (2, 1);
\draw [gray, fill=gray] (0, 0) to[bend right=30] (2, 0);
\draw [gray, fill=gray] (0, 1) to[bend right=30] (2, 1);
\draw [gray, fill=gray] (0, 0) to[bend left=30] (2, 0);

\draw [ultra thick] (0, 0) -- (0, 1);
\draw [ultra thick] (2, 0) -- (2, 1);

\draw [fill=black] (1, 1) circle [radius=0.1];
\draw [fill=black] (1, 0) circle [radius=0.1];
\draw [thick] (0, 0) -- (2, 0);
\draw [thick] (0, 1) -- (2, 1);

\node at (-0.4, 0) {$u^*$};
\node at (-0.4, 1) {$v^*$};
\node at (2.4, 0) {$x$};
\node at (2.4, 1) {$y$};
\node at (2.8, 0.5) {$e' \in F^*$};
\node at (1, 2) {$O(f)$ $v^* \leadsto y$ $2$-paths};
\node at (1, -1) {$O(f)$ $u^* \leadsto x$ $2$-paths};
\node at (-1.5, 0.5) {$e^*$ (latest)};

\end{tikzpicture}
\caption{\label{fig:3meetgen} The four choices in our generation of SCUM $3$-path meets lead to a bound of $O(ndf^3)$ on the total number of meets.  The common endpoints of the two paths in the meet occur as the middle nodes of the $v^* \leadsto y$ and $u^* \leadsto x$ endpoints that we select.}
\end{figure}

\section{Limiting SCUM Path Meets for General Odd \texorpdfstring{$k$}{k} \label{sec:limit}} 

In this section we prove Lemma~\ref{lem:oddmeetbound}, which as discussed implies Theorem \ref{thm:k-upper-main}.

We follow the same rough strategy as for the $k=3$ case.
To give intuition, let us briefly sketch counting of meets between \emph{internally vertex-disjoint} SCUM $k$-paths $\{\pi^\ell, \pi^r\}$.
We can generate such a meet by the following process (see Figure \ref{fig:kmeetcount}):
\begin{enumerate}
\item First choose an edge $e^* = (u^*, v^*)$ to act as the heaviest edge in the meet.
Suppose $e^* \in \pi^\ell$.

\item Similar to the $k=3$ case, we can argue that there exists an edge $e' \in \pi^r$ that is blocked with $e^*$.
For $k=3$ we specifically had that $e'$ occurs in the \emph{middle} position of $\pi^r$; for general $k$, we can only say that it occurs in position $2$ through $k-1$. 
Choose this edge $e'$ and its position in $\pi^r$.

\item Extend $e'$ into a subpath of $\pi^r$ that contains all but the first and last edge.
Let $x, y$ be the endpoints of this $(k-2)$-length subpath of $\pi^r$.

\item Complete the meet by choosing paths of length $(k+1)/2$ from $u^*$ to $x$ and from $v^*$ to $y$.
Limiting the number of such paths requires a technical lemma, but we can prove a reasonably good upper bound by  exploiting the structure of the double-blocking set.
\end{enumerate}

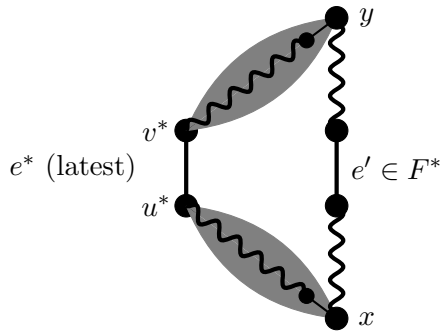
\begin{figure} [h] \centering
\begin{tikzpicture}
\draw [fill=black] (0, 0) circle [radius=0.15];
\draw [fill=black] (0, 1) circle [radius=0.15];

\draw [fill=black] (2, 0) circle [radius=0.15];
\draw [fill=black] (2, 1) circle [radius=0.15];

\draw [gray, fill=gray] (0, 1) to[bend left=30] (2, 2.5);
\draw [gray, fill=gray] (0, 0) to[bend right=30] (2, -1.5);
\draw [gray, fill=gray] (0, 1) to[bend right=30] (2, 2.5);
\draw [gray, fill=gray] (0, 0) to[bend left=30] (2, -1.5);

\draw [ultra thick] (0, 0) -- (0, 1);
\draw [ultra thick] (2, 0) -- (2, 1);

\draw [ultra thick, snake it] (2, 1) -- (2, 2.5);
\draw [ultra thick, snake it] (2, 0) -- (2, -1.5);

\draw [fill=black] (2, 2.5) circle [radius=0.15];
\draw [fill=black] (2, -1.5) circle [radius=0.15];

\draw [fill=black] (1.6, 2.2) circle [radius=0.1];
\draw [fill=black] (1.6, -1.2) circle [radius=0.1];

\draw [thick] (2, 2.5) -- (1.6, 2.2);
\draw [thick] (2, -1.5) -- (1.6, -1.2);

\draw [ultra thick, snake it] (0, 1) -- (1.6, 2.2);
\draw [ultra thick, snake it] (0, 0) -- (1.6, -1.2);

\node at (-0.4, 0) {$u^*$};
\node at (-0.4, 1) {$v^*$};
\node at (2.4, 2.5) {$y$};
\node at (2.4, -1.5) {$x$};
\node at (2.8, 0.5) {$e' \in F^*$};
\node at (-1.5, 0.5) {$e^*$ (latest)};

\end{tikzpicture}
\caption{\label{fig:kmeetcount} A sketch of an argument that counts the \emph{internally vertex-disjoint} SCUM $k$-path meets in $H$.}
\end{figure}

\paragraph{Challenges in extension to general $k$.} This sketched argument is the natural analog of our argument for $k=3$.
In fact, \emph{if it were true that all meets are internally vertex disjoint}, this argument would provide a range of free fault tolerance for both odd and even $k>3$.
Unfortunately, we cannot make such an assumption when $k>3$, and we also need to count meets that nontrivially coincide.
Counting meets that share prefixes/suffixes introduces a new term in our dispersion lemma, which harms our bounds for general odd $k$ relative to $k=3$, and which spoils our ability to prove free fault tolerance for any even $k$.

\subsection{Technical Lemmas on Path Structure} \label{sec:oddtech}

In the following, we say that a path $\pi$ in $H$ is \emph{blocked} if there is an edge pair $\{e_1, e_2\} \in B$ with $e_1, e_2 \in \pi$, and $\pi$ is \emph{double-blocked} if there are two such edge pairs in $B$.
Note that this differs slightly from the notion of a cycle being blocked, as we do not require that $e_1$ is the latest edge in $\pi$.

\begin{lemma} [Extends Lemma \ref{lem:3midedge}] \label{lem:faulthold}
Let $\pi_1, \pi_2$ be simple paths with the same endpoints $s, t$.
Suppose $e^* \in \pi_1$ is the last edge in $\pi_1 \cup \pi_2$ in the ordering, and let $F^*$ be the set of edges that are in a block with $e^*$ in $B$.
Then there exists an edge $e' \in F^* \cup \{e^*\}$ with $e' \in \pi_2$ if either of the following conditions hold:
\begin{itemize}
\item $\pi_1, \pi_2$ are both unblocked and have length $\le k$, or

\item $\pi_1, \pi_2$ are both un-double-blocked and have length $\le (k+1)/2$.
\end{itemize}
\end{lemma}
\begin{proof}
Assume $\pi_1, \pi_2$ are unblocked and have length $\le k$, and suppose $e^* \notin \pi_2$.
Then there is a subpath of $\pi_1$ and a subpath of $\pi_2$ that form a cycle $C$ with $|C| \le 2k$ edges, and latest edge $e^*$.
This cycle must be blocked by $B$; since $e^*$ is its latest edge, there is a block of the form $\{e^*, e'\} \in B$ with $e' \in C$.
Since $\pi_1$ is unblocked, we have $e' \notin \pi_1$, so $e' \in \pi_2$, completing the proof.

The case where $\pi_1, \pi_2$ are un-double-blocked and have length $\le (k+1)/2$ is essentially identical, noting that their subpaths form a cycle of length $\le k+1$ which is thus double-blocked by $B$.
\end{proof}

\begin{lemma} [Extends Lemma \ref{lem:2pathlim}] \label{lem:shortpathcount}
For any two nodes $s, t$ and any $1 \le L \le (k+1)/2$, there are at most
$O\left( f^{L - 1} \right)$
simple $s \leadsto t$ paths of length $L$ that are not double-blocked.
\end{lemma}
\begin{proof}
We prove this lemma by strong induction on $L$.
The base case $L=1$ is trivial, since there is at most $1$ edge between any two nodes.

For the inductive step: let $\Pi_{s, t}$ be the set of all simple $s \leadsto t$ paths of length $L$ that are not double-blocked, let $e^*$ be the latest edge in the ordering contained in any path in $\Pi_{s, t}$, and let $F^*$ be the edges that appear with $e^*$ in a block in $B$.
By the previous lemma, every path $\pi \in \Pi_{s, t}$ contains an edge in $F^* \cup \{e^*\}$.

Fix $(u, v) \in F^* \cup \{e^*\}$ and a position $i$, and let us count the number of paths in $\Pi_{s, t}$ that use $(u, v)$ in position $i$.
\begin{itemize}
\item If $i \in \{2, \dots, L-1\}$, then we may apply the inductive hypothesis twice to argue that we have $O(f^{i-2})$ $s \leadsto u$ paths of length $i-1$, and $O(f^{L-i-1})$ $v \leadsto t$ paths of length $L-i$, which gives $O(f^{L-3})$ paths that use $(u, v)$ as their $i^{th}$ edge.
\item If $i=1$ (and so $s=u$), then we may apply the inductive hypothesis once to argue that we have $O(f^{L-2})$ $v \leadsto t$ paths of length $L-1$.
\item If $i=L$ (and so $v=t$), then we may apply the inductive hypothesis once to argue that we have $O(f^{L-2})$ $s \leadsto u$ paths of length $L-1$.
\end{itemize}
Unioning over the $\le k(f+1)$ choices of $(u, v)$ and $i$, we count $O(f^{L-1})$ paths in $\Pi_{s, t}$.
\end{proof}
The following lemma is new to general $k$ and does not extend any previous lemmas.
\begin{lemma} \label{lem:doubleflexmeet}
Let $s, t$ be nodes, let $0 \le L_1 \le (k-1)/2$, and let $L_2 \ge 1$.
Then there are $O(f^{L_1} d^{L_2-1})$ pairs of simple paths $(\pi_1, \pi_2)$ in $H$ with the following properties:
\begin{itemize}
\item $\pi_1$ has length $L_1$ and start point $s$. $\pi_2$ has length $L_2$ and start point $t$.
$\pi_1, \pi_2$ have the same endpoint node
(that is: we are counting pairs of $s \leadsto x, t \leadsto x$ paths, over all possible choices of node $x$).

\item $\pi_1$ is chain-unblocked.
\item $\pi_2$ is not a subpath of $\pi_1$.
\end{itemize}
\end{lemma}
\begin{proof}
First, we acknowledge that it is possible that the paths in the meet $\pi_1, \pi_2$ coincide on a nontrivial suffix.
Let $c$ be the number of edges in the longest common suffix of $\pi_1, \pi_2$.
We begin the proof by assuming $c=0$; we will revisit the general case at the end.
A pair of paths $(\pi_1, \pi_2)$ as described in the lemma may be generated by the following process.
\begin{itemize}
\item Starting from $t$, choose any simple path of length $L_2-1$.
We know from Section~\ref{sec:cleaning} that without loss of generality we may assume that there are $\Theta(d)$ edges incident to each node, and so there are $O(d)^{L_2-1}$ possible choices for this path.

\item Since we have assumed $c=0$, we notice that $\pi_1$ plus the last edge of $\pi_2$ forms a simple path $q$ of length $L_1+1 \le (k+1)/2$.
The last pair of edges on $q$ (i.e., the last edge of $\pi_1$ and the last edge of $\pi_2$) could possibly appear as a block in $B$.
However, since $\pi_1$ is chain unblocked, no other pair of edges on $q$ can appear as a block.
Thus, $q$ is not double-blocked.
Applying Lemma \ref{lem:shortpathcount}, there are $O(f^{L_1})$ ways to complete the rest of the meet.
\end{itemize}
We count $O(d)^{L_2-1} \cdot O(f^{L_1})$ ways to generate this meet, which gives our claimed bound for $c=0$.
We then consider general $c$.
Notice that a pair of paths $(\pi_1, \pi_2)$ as described in the lemma, which coincide on their last $c$ edges, may be identified by (1) a pair of paths as described in the lemma of lengths $L_1-c, L_2-c$, which do not share a nontrivial suffix, and (2) an additional path of length $c$ from their shared endpoint, representing the common suffix.
Since we assume that $\pi_2$ is not a subpath of $\pi_1$, we note that $L_2 - c \ge 1$.
So we may apply the first part of the argument, and conclude that there are $O(f^{L_1-c} d^{L_2-c-1})$ ways to choose the path pair of lengths $L_1-c, L_2-c$, and then there are $O(d)^c$ ways to choose the common suffix.
This gives a stronger bound of $O(f^{L_1-c} d^{L_2-1})$ on the number of path pairs in the setting of general $c$.
\end{proof}

\subsection{Proof of Lemma~\ref{lem:oddmeetbound}} \label{sec:meetboundodd}

We count SCUM $k$-path meets $\{\pi^\ell, \pi^r\}$ by the following process.
\begin{itemize}
\item Choose an edge $e^* = (u^*, v^*)$ to act as the latest edge in $\pi^{\ell} \cup \pi^r$ in the ordering.
The latest edge in the ordering must be the middle ($m=(k+1)/2$) edge of either $\pi^\ell$ or $\pi^r$; we will assume without loss of generality that $e^*$ is the middle edge of $\pi^\ell$.
There are $O(nd)$ possible choices of $e^*$.

\item Let $F^*$ be the edges that are paired with $e^*$ in $B$.
By Lemma \ref{lem:faulthold}, $\pi^r$ contains an edge $e' = (u', v') \in F^* \cup \{e^*\}$.
Moreover, since $\pi^\ell$ is chain unblocked, $e'$ cannot be the first or last edge in $\pi^r$; rather, it occurs in position $2$ through $k-1$.
Select one of these positions for $e'$ along $\pi^r$.
There are $\le (f+1)(k-2)$ possible choices of $e'$ and its position in $\pi^r$.

\item We next select prefixes for $\pi^\ell, \pi^r$; that is, we determine the parts of the paths respectively preceding nodes $u^*, x$.
We can view the prefix of $\pi^\ell$ as an oriented SCU path of length $m-1$, and the prefix of $\pi^r$ as a simple path of length $i-1 \ge 1$.
There are two cases:
\begin{itemize}
\item If the suffix of $\pi^r$ is not a subpath of the suffix of $\pi^\ell$, then the conditions of Lemma \ref{lem:doubleflexmeet} hold.
Applying this lemma, we have $O(f^{m-1}d^{i-2})$ choices of prefix pairs.

\item If the prefix of $\pi^r$ is a subpath of the prefix of $\pi^\ell$, then it is fully determined by a choice of prefix of $\pi^\ell$.
There are $O(d)^{m-1}$ possible choices of prefix.
\end{itemize}

\item Finally, we select suffixes for $\pi^\ell, \pi^r$.
Exactly like the prefixes, we can view the suffix of $\pi^\ell$ as an oriented SCU path of length $m-1$, and the suffix of $\pi^r$ as a simple path of length $k-i \ge 1$.
We have the same two cases in our counting:
\begin{itemize}
\item In the case where the suffix of $\pi^r$ is not a subpath of the suffix of $\pi^\ell$, then by Lemma \ref{lem:doubleflexmeet} we have $O(f^{m-1}d^{k-i-1})$ choices of suffix pairs.

\item In the case where the suffix of $\pi^r$ is a subpath of the suffix of $\pi^\ell$, then we have $O(d)^{m-1}$ choices of suffix.
\end{itemize}
\end{itemize}

We now put the parts together.
There are $O(ndf)$ ways to make choices in the first two steps.
If neither the prefix nor suffix of $\pi^r$ is a subpath of the prefix/suffix of $\pi^\ell$, then the number of ways to choose prefixes and suffixes is
$$ O\left(f^{m-1}d^{i-2}\right) \cdot O\left(f^{m-1} d^{k-i-1}\right) = O\left(f^{2m-2} d^{k-3}\right).$$
If the prefix of $\pi^r$ is a subpath of the prefix of $\pi^\ell$, but the suffix is not, then the number of choices is
$$ O\left(d^{m-1}\right) \cdot O\left(f^{m-1} d^{k-i-1}\right) = O\left(f^{m-1} d^{k-(i-m)-2}\right) = O\left(f^{m-1} d^{k-2}\right),$$
where the last equality is because we must have $i \le m$ since the prefix of $\pi^r$ is a subpath of the prefix of $\pi^\ell$.
The counting is symmetric in the case where the suffix of $\pi^r$ is a subpath of the suffix of $\pi^\ell$, and leads to the same bound.
Finally, we do not need to consider the case where both the prefix and suffix of $\pi^r$ is a subpath of the prefix/suffix of $\pi^\ell$, as this would imply $\pi^\ell = \pi^r$ but the paths in a meet must be distinct.

The claimed bound on meets in this lemma now follows by multiplying our prefix/suffix bounds in each case by $ndf$, and substituting $m = (k+1)/2$.

\section{Polynomial-Time Version} \label{sec:polytime}
Algorithm~\ref{alg:greedy} is actually slightly underspecified, since we do not describe how to determine whether there is some $F \subseteq E$ with $|F| \leq f$ such that $\dist_{H^F}(u,v) > (2k-1) \cdot w(u,v)$.  The obvious approach simply tries all possible $F$.  This is correct, but takes time at least $\Omega\left(\binom{m}{f}\right)$, which is not polynomial when $f$ is not constant.  Hence while Algorithm~\ref{alg:greedy} is sufficient to prove the structural results we claim, it is not actually \emph{efficient} as an algorithm.

This parallels the situation with fault-tolerant spanners, where the first truly tight bounds on the size used a similar greedy algorithm that did not run in polynomial time~\cite{BP19}.  To overcome this, Dinitz and Robelle \cite{DR20} showed that slightly suboptimal bounds (by a factor of $k$) could be achieved by \emph{approximating} the ``if'' condition rather than actually solving it.  We now show that a similar result can be obtained for edge-fault-tolerant emulators: we can modify Algorithm~\ref{alg:greedy} to run in polynomial time at the cost of only $O(k)$ in the achieved size.

\subsection{Length-Bounded Double-Cut}
The problem that~\cite{DR20} had to approximate to give a polynomial-time check in the ``if'' condition was known as \emph{Length-Bounded Cut}.  Since fault-tolerant emulators act quite differently than fault-tolerant spanners, this is unfortunately not the appropriate problem for us to approximate.  In particular, for spanners the existence of a small length-bounded cut implied that it was necessary to add the edge to the solution, but for emulators this is not true.  So we need to find a different condition to check.  It turns out that, for essentially the same reason that our emulators had small \emph{double}-blocking sets, the right problem for us to solve is the following.
\begin{definition}
In the \emph{Length Bounded Double Cut} problem (LBDC), we are given an unweighted graph $G = (V, E)$, two special nodes $u,v \in V$, and two integers $t > \ell$.  A feasible solution is a set $F \subseteq E$ such that the following two properties hold:
\begin{enumerate}
    \item (Long Paths) Every path $P$ from $u$ to $v$ with at most $t$ edges has $|P \cap F| \geq 1$, and
    \item (Short Paths) Every path $P$ from $u$ to $v$ with at least $2$ and at most $\ell$ edges has $|P \cap F| \geq 2$.
\end{enumerate}
The objective is to find a feasible $F$ minimizing $|F|$.
\end{definition}

We now want to show that LBDC can be approximated to within $O(t)$.  First, we prove a useful lemma about testing.

\begin{lemma} \label{lem:LBDC-test}
There is a polynomial-time algorithm which, given an instance $(G = (V, E), u, v, t, \ell)$ of LBDC and $F \subseteq E$, returns YES if $F$ is a feasible solution and otherwise returns a path $P$ which violates either the Long Paths or Short Paths condition.
\end{lemma}
\begin{proof}
We can test whether the long path condition is satisfied by running BFS from $u$ to $v$ in $G \setminus F$ and checking whether there is a path $P$ of length at most $t$.  If there is such a path then $F$ is not feasible and we return $P$, but if no such path exists then $F$ hits all long paths as required.

We then test the short path condition, which is only slightly more complicated.  We need to determine whether there is any path $P$ from $u$ to $v$ with at most $\ell$ edges and which uses at most one edge of $F$.  There cannot be any such path which has no edges of $F$, since then it would have length at most $\ell < t$ and so would have been found in the long path part of the algorithm and already returned.  So without loss of generality, any such $P$ has exactly one edge in $F$.  To find such a $P$, we can just test for every edge in $F$ whether it is part of such a path. 

Slightly more formally, for every edge $e = (x,y) \in F$, we compute four shortest paths in $G \setminus F$: we let $P_{ux}$ be the shortest $u-x$ path, $P_{yv}$ be the shortest $y-v$ path, $P_{uy}$ be the shortest $u-y$ path, and $P_{xv}$ be the shortest $x-v$ path.  Obviously all of these can be computed in polynomial time.  Then we test whether $|P_{ux}| + 1 + |P_{yv}| \leq \ell$, in which case we have found a violated path to return ($P_{ux}$ followed by $(x,y)$ followed by $P_{yv}$).  Similarly, we test whether $|P_{uy}| + 1 + |P_{xv}| \leq \ell$, in which case we return the appropriate path.  If we do not return any path for ant $e \in F$, then no path violating the short paths condition exists, so we can return YES.
\end{proof}

We can now use this to approximate LBDC:

\begin{theorem} \label{thm:LBDC}
Given an instance $(G = (V, E), u, v, t, \ell)$, there is a polynomial-time $t$-approximation algorithm for LBDC.
\end{theorem}
\begin{proof}
Our algorithm is simple to state now that we have Lemma~\ref{lem:LBDC-test}.  Initially we set $F = \emptyset$. First we use Lemma~\ref{lem:LBDC-test} to test whether $F$ is feasible.  If it is feasible then we return it.  Otherwise, we take the path $P$ returned by Lemma~\ref{lem:LBDC-test} and add all of its edges to $F$.  We repeat until $F$ is feasible.

It is trivial to see that this algorithm returns a feasible solution: if $F$ is not feasible then it will not yet return $F$.  It is also trivial to see that this runs in polynomial time, since in every iteration it adds at least one new edge to $F$.  So it remains only to bound its approximation ratio.

Let $F^*$ be the optimal solution.  We claim that in every iteration of the algorithm, the path $P$ that it finds (and adds to $F$) has at least one edge in $F^*$ that is not already in $F$.  To see this, note that $P$ broke either the long path or the short path condition for the current algorithmic solution $F$ (since the algorithm of Lemma~\ref{lem:LBDC-test} returned it).  If it broke the long path condition, then it has length at most $t$ and did not intersect with $F$.  But since $F^*$ is feasible we know that $|P \cap F^*| \geq 1$, so $P$ has at least one edge in $F^*$ that is not in $F$.  Similarly, if it broke the short path condition, then it has length at most $\ell$ and has $|P \cap F| \leq 1$.  But since $F^*$ is feasible we know that $|P \cap F^*| \geq 2$, and hence there is at least one edge of $F^* \setminus F$ in $P$.

Thus in every iteration we add at least one edge of $F^*$ to $F$ that was not already in $F$.  Thus there are at most $|F^*|$ iterations.  In each iteration we add every edge of a path of length at most $t$, and hence $|F| \leq t \cdot |F^*|$.  Thus the algorithm is a $t$-approximation.  
\end{proof}

\subsection{Polynomial-Time EFT Emulators}
We can now design and analyze our algorithm.  It is almost identical to Algorithm~\ref{alg:greedy}, but using Theorem~\ref{thm:LBDC} in the if condition.  

\begin{algorithm} [ht]
\textbf{Input:} Graph $G = (V, E, w)$, positive integers $f, k$\;
~\\

Let $H \gets (V, \emptyset)$ be the initially-empty emulator\;
\ForEach{edge $(u, v) \in E$ in order of nondecreasing weight $w(u, v)$}{
    Let $F$ be the edge set returned by the algorithm of Theorem~\ref{thm:LBDC} when run on input $(H, u,v, 2k-1, k)$ (note that this interprets $H$ as an unweighted graph)\;
    \If{$|F| \le (2k-1)f$}{
        Add $(u, v)$ to $H$\;
    }
}
\Return{$H$};
\caption{Polynomial-time Algorithm for $f$-EFT $(2k-1)$-emulators}
\label{alg:poly}
\end{algorithm}

This algorithm obviously takes polynomial time.  So we need to prove that the emulator $H$ it returns is a feasible $f$-EFT $(2k-1)$-emulator, and that it has the correct size.  This is quite similar to our original proof that the non-polytime algorithm gave a solution with a small double-blocking set (Lemma~\ref{lem:doubleblocking}), but for Algorithm~\ref{alg:poly} this argument shows up in the correctness proof.

\begin{lemma}
$H$ is an $f$-EFT $(2k-1)$-emulator of $G$.
\end{lemma}
\begin{proof}
Let $F$ be an arbitrary set of $|F| \leq f$ failing edges.  We need to show that $\dist_{H^F}(u,v) \leq (2k-1) \cdot w(u,v)$ for all edges $(u,v) \in E \setminus F$.  There are two cases.
\begin{itemize}
    \item First, if $(u,v) \in H \setminus F$, then we immediately have $\dist_{H^F}(u,v) = \dist_{G \setminus F}(u,v) \leq w(u,v)$, due to the reweighting of $(u,v)$ in $H$.
    \item Second, suppose $(u,v) \not\in H \setminus F$.  Since $(u,v) \in E \setminus F$, this means that $(u,v) \not\in H$.  Therefore we must have considered the edge $(u,v)$ at some point in the algorithm, and then decided not to add it to $H$.  Fix $H$ at this point in time, and suppose for contradiction that $\dist_{H^F}(u,v) > (2k-1) \cdot w(u,v)$.  
    
    We claim that $F$ must be a feasible solution to the LBDC problem on input $(H, u, v, 2k-1, k)$.   To see this, first consider the long paths condition.  If there was a path $P$ from $u$ to $v$ with at most $2k-1$ edges with $P \cap F = \emptyset$, then all of these edges still exist in $H^F$, and since we consider edges in greedy order they all have length at most $w(u,v)$, and hence $\dist_{H^F}(u,v) \leq (2k-1) \cdot w(u,v)$.  This contradicts our assumption on $F$, and hence no such path exists.  On the other hand, suppose that there is a path $P$ with at most $k$ edges from $u$ to $v$ with $|P \cap F| = 1$.  Let $e' = (x,y)$ be the unique edge in $P \cap F$.  Then since all other edges of $P$ are not in $F$, and $(u,v) \not \in F$, we know by the definition of reweighting in $H^F$ that 
    \begin{align*}
        \dist_{H^F}(x,y) &\leq w(u,v) + \sum_{e'' \neq e' \in P} w(e'') \leq k \cdot w(u,v)
    \end{align*}
    This implies that
    \begin{align*}
        \dist_{H^F}(u,v) &\leq \dist_{H^F}(x,y) + \sum_{e'' \neq e' \in P} w(e'') \leq (2k-1) \cdot w(u,v),
    \end{align*}
    which contradicts our assumption on $F$.  Thus no such short path exists, and so $F$ is a feasible LBDC solution.
    
    Since $F$ is a feasible LBDC solution on in put $(H,u,v,2k-1,k)$ and Theorem~\ref{thm:LBDC} gives a $(2k-1)$-approximation for this instance, we know that the set $\hat F$ it returned has $|\hat F| \leq (2k-1) |F| \leq (2k-1)f$.  But this means that the algorithm would have added $(u,v)$ to $H$.  This is a contradiction, and hence $\dist_{H^F}(u,v) \leq (2k-1) \cdot w(u,v)$. \qedhere
\end{itemize}
\end{proof}

Now we show that emulator $H$ returned by Algorithm~\ref{alg:poly} has a small double-blocking set.

\begin{lemma}
The emulator $H$ returned by Algorithm~\ref{alg:poly} has a $2k$ double-blocking set $B$ of size $|B| \leq (2k-1) f |E(H)|$ (with respect to the order that the edges are added to $H$ in the algorithm).
\end{lemma}
\begin{proof}
For every edge $(u,v)$ the algorithm adds to $H$, let $F_{(u,v)}$ be the set of edges of size at most $(2k-1)f$ which caused the algorithm to add $(u,v)$, i.e., which was returned by the algorithm of Theorem~\ref{thm:LBDC}.  Note in particular that $F_{(u,v)}$ is a feasible length-bounded double-cut for the instance $(H, u,v, 2k-1, k)$ (where $H$ in this instance consists of the edges of the returned emulator that were considered before $(u,v)$ by the algorithm).  Define
\[
B := \{\{e,f\} : e \in E(H), f \in F_e\}.
\]
Since $|F_e| \leq (2k-1)f$ for all $e \in H$, we have that $|B| \leq (2k-1) f |E(H)|$ as required.  We now show that $B$ is a $2k$ double-blocking set.
\begin{itemize}
    \item Let $C$ be a cycle in $H$ with at most $2k$ edges, and let $(u,v) \in C$ be the last edge of $C$ added by the algorithm.  When we add $(u,v)$ to $H$, the other edges in $C$ form a $u \leadsto v$ path with at most $2k-1$ edges.  Since $F_{(u,v)}$ is a feasible LBDC solution, this implies that $|F_{(u,v)} \cap C| \geq 1$.  Let $f \in F_{(u,v)} \cap C$.  Then $\{(u,v), f\} \in B$, as required.
    \item Let $C$ be a cycle in $H$ with at most $k+1$ edges, and let $(u,v) \in C$ be the last edge of $C$ added by the algorithm.   When we add $(u,v)$ to $H$, the other edges in $C$ form a $u \leadsto v$ path with at most $k$ edges.  Since $F_{(u,v)}$ is a feasible LBDC solution, this implies that $|F_{(u,v)} \cap C| \geq 2$.  Let $e, f \in F_{(u,v)} \cap C$.  Then $\{(u,v),e\} \in B$ and $\{(u,v), f\} \in B$ by the definition of $B$, as required. \qedhere
\end{itemize}
\end{proof}

As Theorems~\ref{thm:3-blocked-upper} and~\ref{thm:k-blocked-upper} bound the size of any graph with a double-blocking set, their edge bounds therefore apply to the output emulator, just as they applied to the output emulator of Algorithm~\ref{alg:greedy}.  Since this double-blocking set is only $(2k-1)$ times larger than the double blocking set proved in Lemma~\ref{lem:doubleblocking} for Algorithm~\ref{alg:greedy}, the total increase in the final edge bound we get is only $O(k)$.

\section{Lower Bounds} \label{sec:lower}

\subsection{General Lower Bound}

\begin{theorem}
For all $k, f$, there are $n$-node graphs $G$ for which any $f$-EFT $k$-emulator $H$ has $|E(H)| = \Omega(nf)$.
\end{theorem}
\begin{proof}
Let $G$ be any $f/2$-regular $n$-node graph, and let $H$ be an $f$-EFT $k$-emulator of $G$.
Consider an edge $(u, v) \in E(G)$, and seeking contradiction, suppose $(u, v) \notin E(H)$.
Consider the fault set $F$ that contains all edges incident to $u$ and all edges incident to $v$, except for $(u, v)$ itself.
We have $\dist_{G \setminus F}(u, v) = 1$, and also for all nodes $x \notin \{u, v\}$, we have
$$\dist_{G \setminus F}(u, x) = \dist_{G \setminus F}(v, x) = \infty.$$
Thus, any edge incident to $u$ or $v$ in $H^F$ has weight $\infty$, and so $\dist_{H^F}(u, v) = \infty$.
Thus $H$ is not a $k$-emulator of $G \setminus F$, completing the contradiction.
\end{proof}

\subsection{Lower Bound for \texorpdfstring{$k=2$}{k=2}}

Here we prove the following lower bound for the setting $k=2$, i.e., emulators with stretch $3$:
\begin{theorem}
For any $f$, there are $n$-node graphs $G$ for which any $f$-EFT $3$-emulator $H$ has $|E(H)| = \Omega(f^{1/2} n^{3/2})$.
\end{theorem}

We use a standard construction in prior work \cite{BDPW18, BDN22}.
Start with a graph on $n/f$ nodes,\footnote{We ignore issues of non-integrality, which affect our bounds only by lower-order terms.} $\Omega(n/f)^{3/2}$ edges, and girth $>4$.
It will be a little convenient later to assume this graph is bipartite (without loss of generality).
Then, blow up each node $v$ into a ``cloud" $C_v$, which contains $f$ copies of $v$.
For each edge $(u, v)$ we include an edge between every copy in $C_u$ and every copy in $C_v$.
This graph has $\Omega(f^{1/2} n^{3/2})$ edges.
Let $H$ be an $4f$-EFT $3$-emulator of $G$ (with a sufficiently large implicit constant hidden in the $O$).
We will prove that $|E(H)| = \Omega(f^{1/2} n^{3/2})$.

Consider an edge $(u, v) \in E(G)$ and define the following fault set $F$, whose size $|F|$ might be much larger than $4f$.
See Figure \ref{fig:k=3lbfaults} for a picture.
\begin{itemize}
    \item Fault all edges from $u$ to any node in the cloud $C_v$, and fault all edges from $v$ to any node in the cloud $C_u$, except for $(u, v)$ itself.
    This contributes $2f-1$ edges to $F$.
    
    \item For each edge $(u, x) \in E(H)$ for which $u, x$ lie in different clouds and there exists a $2$-path in $G$ of the form $(u, v, x)$, fault the edge $(v, x)$.
    Let $Z_u \subseteq F$ be all edges faulted in this way.
    
    \item For each edge $(t, v) \in E(H)$ for which there exists a $2$-path of the form $(t, u, v)$, fault the edge $(t, u)$.
    Let $Z_v \subseteq F$ be all edges faulted in this way.
\end{itemize}

\begin{figure} [ht] \centering
\begin{tikzpicture}
\draw (0,0) ellipse (0.5 and 2);
\draw (2,0) ellipse (0.5 and 2);
\node at (0, -2.5) {$C_u$};
\node at (2, -2.5) {$C_v$};

\draw [fill=black] (0, 0) circle [radius=0.2];
\draw [fill=black] (0, 1) circle [radius=0.15];
\draw [fill=black] (0, -1) circle [radius=0.15];

\draw [fill=black] (2, 0) circle [radius=0.2];
\draw [fill=black] (2, 1) circle [radius=0.15];
\draw [fill=black] (2, -1) circle [radius=0.15];

\node at (0, -0.4) {$u$};
\node at (2, -0.4) {$v$};

\draw [line width=3] (0, 0) -- (2, 0);
\draw [red, ultra thick, dotted] (0, 0) -- (2, 1);
\draw [red, ultra thick, dotted] (0, 0) -- (2, -1);
\draw [red, ultra thick, dotted] (0, 1) -- (2, 0);
\draw [red, ultra thick, dotted] (0, -1) -- (2, 0);

\draw [fill=black] (-2, 1) circle [radius=0.2];
\draw [fill=black] (4, 1) circle [radius=0.2];
\node at (-2, 0.6) {$t$};
\node at (4, 0.6) {$x$};

\draw [line width = 3] (-2, 1) to[bend left = 5] (2, 0);
\draw [ultra thick, blue, dotted] (-2, 1) -- (0, 0);

\draw [line width = 3] (0, 0) to[bend left = 5] (4, 1);
\draw [ultra thick, yellow!50!black, dotted] (2, 0) -- (4, 1);

\end{tikzpicture}
\caption{\label{fig:k=3lbfaults} The dotted lines in this diagram are edges in the fault set $F$ associated to $(u, v)$, and the solid black lines are edges in $H$ that are not contained in $F$.  The dotted blue edge on the left is in $Z_v$, and the dotted yellow edge on the right is in $Z_u$.  For clarity, we have omitted some (unfaulted) edges in this picture going between $C_u$ and $C_v$.}
\end{figure}
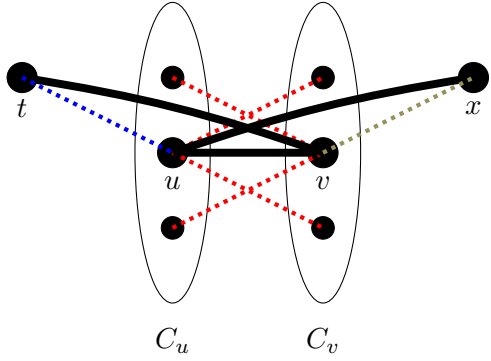

Recall that $H^F$ is the reweighted version of $H$ under this fault set $F$.
We will describe edges in $H^F$ as ``spanner edges'' if the edge is also in $G \setminus F$, or ``emulator edges'' if not.

\begin{claim}
If $(u, v) \notin E(H)$, then $\dist_{H^F}(u, v) > 3$.
\end{claim}
\begin{proof}
We will prove this claim in three cases.
We argue that there is no one-edge $u \leadsto v$ path in $H^F$, there is no \emph{short} two-hop $u \leadsto v$ path in $H^F$, and there is no \emph{short} three-hop $u \leadsto v$ path in $H^F$.
Since all edge weights in $H^F$ are at least $1$, any path with $\ge 4$ hops has length $\ge 4$ and may be ignored.
\begin{itemize}
    \item (One-hop paths) By assumption, $(u, v) \notin E(H)$, so there is no $1$-hop path from $u$ to $v$.
    
    \item (Two-hop paths) We next argue that there is no $2$-hop path from $u$ to $v$ of length $\le 3$.
    \begin{itemize}
    \item Since $G$ is bipartite and $(u, v) \in E(G)$, there can be no $2$-hop $u \leadsto v$ path in $H^F$ that specifically uses two \emph{spanner} edges.
    
    \item Each emulator edge in $H^F$ has weight at least $2$, so any $2$-hop $u \leadsto v$ path in $H^F$ that uses two \emph{emulator} edges has length $\ge 4$.
    
    \item The remaining subcase is to argue that there is no $2$-hop $u \leadsto v$ path of total length $\le 3$ in $H^F$ that uses exactly one spanner edge and one emulator edge.
    Suppose towards contradiction that such a path $\pi$ exists, and note that its emulator edge needs to have weight exactly $2$.
    Since $G$ has girth $>4$, the path $\pi$ cannot correspond to a $u \leadsto v$ $3$-path in $G$ that avoids the edge $(u, v)$.
    It follows that the emulator edge in $\pi$ either has the form $(u, x)$ and there is a $2$-path of the form $(u, v, x)$ in $G$, or it has the form $(t, v)$ and there is a $2$-path $(t, u, v)$ in $G$.
    In either case, the associated spanner edge $(v, x)$ or $(t, u)$ is included in $Z_u$ or $Z_v$ (respectively).
    Thus the remaining edge in $\pi$ would not be a spanner edge, completing the contradiction.
    \end{itemize}
    
    \item (Three-hop paths) Finally, we argue that there is no $3$-edge path from $u$ to $v$ of length $\le 3$.
    Such a path would need to use all spanner edges, since emulator edges have weight $\ge 2$.
    Since $(u, v) \notin E(H^F)$ and the original graph has girth $>4$, no such $3$-path may exist. \qedhere
\end{itemize}
\end{proof}

\begin{claim}
Either $(u, v) \in E(H)$, $|Z_u| > f$, or $|Z_v| > f$.
\end{claim}
\begin{proof}
Suppose $(u, v) \notin E(H)$.
From the previous claim, we have $\dist_{H \setminus F}(u, v) > 3$.
Since $H$ is a $4f$-EFT $3$-emulator, this is possible only if $|F| > 4f$, which is possible only if $|Z_u| > f$ or $|Z_v| > f$.
\end{proof}

Now we count edges in $E(H)$.
For each pair of adjacent clouds $C_u, C_v$, there are two cases:
\begin{itemize}
\item Suppose that at least half of the edges in $C_u \times C_v$ are present in $E(H)$.
In this case, this cloud pair contributes $\Omega(f^2)$ edges to $E(H)$.

\item Otherwise, suppose that at least half of the edges in $C_u \times C_v$ are missing in $E(H)$.
Let $(u, v)$ be one such missing edge.
By our previous claim, we have $|Z_u| > f$ or $|Z_v| > f$.
There are two cases:
\begin{itemize}
\item ($|Z_u| > f$) In this case, we have $f$ edges in $E(H)$ of the form $(t, v)$ with a $t \leadsto v$ $2$-path in $G$.
Notice that all such $2$-paths have the form $(t, u', v)$ with $u' \in C_u$.
We may charge all $> f$ of these edges to the pair $(C_u, v)$.

\item ($|Z_v| > f$) In this case, following identical logic to the previous case, we may charge $>f$ edges to the pair $(u, C_v)$.
\end{itemize}
To summarize, at least half of the edges in $C_u \times C_v$ are missing in $E(H)$, and each missing edge goes between two nodes $(u, v)$ where we have charged either $(u, C_v)$ or $(C_u, v)$ for $>f$ emulator edges.
It follows that $\Omega(f^2)$ total emulator edges are charged in total when considering $C_u \times C_v$.
(We remark that each emulator edge is charged twice in total; e.g., once from the cloud pair $C_u \times C_v$ and again from the cloud pair $C_t \times C_u$, but this factor of two may be ignored.)
\end{itemize}

Putting it together: for each of the $\Omega((n/f)^{3/2})$ adjacent cloud pairs in $G$, we charge $\Omega(f^2)$ edges in $H$ (due to one case or the other).
Thus we have
$$|E(H)| = \Omega\left(f^2 \cdot (n/f)^{3/2}\right) = \Omega\left( f^{1/2} n^{3/2} \right).$$

\bibliography{refs}

\newpage
\appendix

\section{Graph Cleaning} \label{app:cleaning}
The following lemma and proof are essentially standard.
\begin{lemma} \label{lem:normalize}
For any $n$-node graph $H$ and double-blocking set $B$ of size $|B| \le f|E(H)|$, there exists a $\Theta(n)$-node graph $H'$ with and double-blocking set $B'$ with the following properties:
\begin{itemize}
\item Every edge in $H$ participates in $\le 2f$ pairs in $B$.
\item Letting $d$ be the average degree in $H$, every node in $H'$ has degree $\Theta(d)$.
\end{itemize}
\end{lemma}
\begin{proof}
We modify $H, B$ as follows.
First: while there exists an edge $e \in E(H)$ that participates in $>2f$ blocks, delete $e$ from $H$ and delete all blocks containing $e$ from $B$.
Since $|B| \le f|E(H)|$, we delete at most $|E(H)|/2$ edges in this process.
When complete, every edge participates in $\le 2f$ pairs in $B$, and the first point holds by reparametrizing $f \gets 2f$.

Second: fixing $d$ as the \emph{current} average degree in $H$, we modify as follows.
While there is a node of degree $\le d/4$, delete that node from $H$.
While there is a node of degree $\ge 2d$, split it into two new nodes, with edges split equitably between the two new nodes.
When this process terminates, it is clear that all remaining nodes have degree $\Theta(d)$.
We delete at most $nd/4 \le |E(H)|/2$ edges due to the first case, so the graph density changes by only a constant factor.
Since density and average degree each change by only a constant factor, it follows that the number of nodes $n$ only changes by a constant factor as well; that is, the final graph has $\Theta(n)$ nodes.
\end{proof}

This lemma implies that without loss of generality, when proving Theorems~\ref{thm:3-blocked-upper} and~\ref{thm:k-blocked-upper}, we may assume that every edge in $H$ participates in at most $2f$ blocks in $B$, and that every node has degree that is $\Theta$ of the average degree.  This is because if we prove Theorems~\ref{thm:3-blocked-upper} and~\ref{thm:k-blocked-upper} under these assumptions, then if there were some other graph $\hat H$ with a $2k$ double-blocking set of size at most $f|E(H)|$ that violated the upper bound of Theorem~\ref{thm:3-blocked-upper} or~\ref{thm:k-blocked-upper}, then Lemma~\ref{lem:normalize} would imply that there is some $H'$ which satisfies the extra assumption but violates the upper bounds of Theorems~\ref{thm:3-blocked-upper} and~\ref{thm:k-blocked-upper}.  This would be a contradiction.   





The assumption that $f \leq d / (Ck)$ is valid because we also pay a term $+fn$ in our final edge bound.  So if $f > d / (Ck)$ then the upper bounds of Theorems~\ref{thm:3-blocked-upper} and~\ref{thm:k-blocked-upper} include a term that is $\Theta(nd/k) = \Theta_k(nd)$, which is trivially satisfied since $|E(H)| = nd/2$ by definition.  


\section{Further Discussion on the Fault-Tolerant Emulator Model} \label{app:emulator-model}

Recall that in our model of fault-tolerant emulators, the length of an edge $\{u,v\}$ in the emulator ``automatically'' updates to the length of the shortest path in the graph post-faults.  That is, for fault set $F$, the length of $\{u,v\}$ is equal to $\dist_{G \setminus F}(u,v)$.  For example, consider Figure~\ref{fig:weightupdate}: under different fault sets, the top emulator edge is assigned different weights.

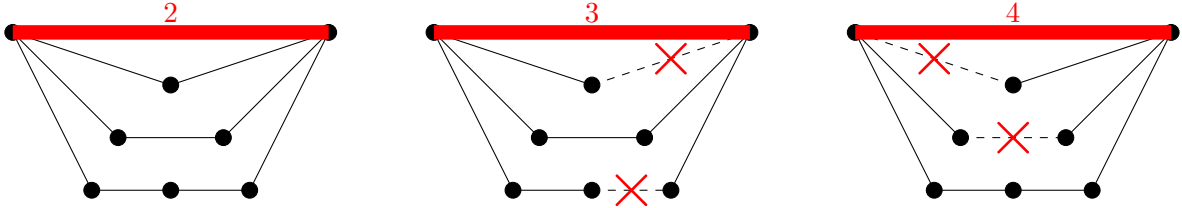
\begin{figure} \centering
\begin{tikzpicture}[scale=0.7]
\draw [fill=black] (0, 0) circle [radius=0.15];
\draw [fill=black] (6, 0) circle [radius=0.15];

\draw [fill=black] (3, -1) circle [radius=0.15];
\draw [fill=black] (2, -2) circle [radius=0.15];
\draw [fill=black] (4, -2) circle [radius=0.15];
\draw [fill=black] (1.5, -3) circle [radius=0.15];
\draw [fill=black] (3, -3) circle [radius=0.15];
\draw [fill=black] (4.5, -3) circle [radius=0.15];

\draw (0, 0) -- (3, -1) -- (6, 0);
\draw (0, 0) -- (2, -2) -- (4, -2) -- (6, 0);
\draw (0, 0) -- (1.5, -3) -- (3, -3) -- (4.5, -3) -- (6, 0);

\draw [line width = 0.5em, red] (0, 0) -- (6, 0);
\node [above] at (3, 0) {\color{red} \bf $2$};

\begin{scope}[shift={(8, 0)}]
\draw [fill=black] (0, 0) circle [radius=0.15];
\draw [fill=black] (6, 0) circle [radius=0.15];

\draw [fill=black] (3, -1) circle [radius=0.15];
\draw [fill=black] (2, -2) circle [radius=0.15];
\draw [fill=black] (4, -2) circle [radius=0.15];
\draw [fill=black] (1.5, -3) circle [radius=0.15];
\draw [fill=black] (3, -3) circle [radius=0.15];
\draw [fill=black] (4.5, -3) circle [radius=0.15];

\draw (0, 0) -- (3, -1);
\draw [dashed] (3, -1) -- (6, 0);
\node at (4.5, -0.5) {\Huge \bf \color{red} $\times$};

\draw (0, 0) -- (2, -2) -- (4, -2) -- (6, 0);

\draw (0, 0) -- (1.5, -3) -- (3, -3);
\draw [dashed] (3, -3) -- (4.5, -3);
\draw (4.5, -3) -- (6, 0);
\node at (3.75, -3) {\Huge \bf \color{red} $\times$};

\draw [line width = 0.5em, red] (0, 0) -- (6, 0);
\node [above] at (3, 0) {\color{red} \bf $3$};

\end{scope}

\begin{scope}[shift={(16, 0)}]
\draw [fill=black] (0, 0) circle [radius=0.15];
\draw [fill=black] (6, 0) circle [radius=0.15];

\draw [fill=black] (3, -1) circle [radius=0.15];
\draw [fill=black] (2, -2) circle [radius=0.15];
\draw [fill=black] (4, -2) circle [radius=0.15];
\draw [fill=black] (1.5, -3) circle [radius=0.15];
\draw [fill=black] (3, -3) circle [radius=0.15];
\draw [fill=black] (4.5, -3) circle [radius=0.15];

\draw [dashed] (0, 0) -- (3, -1);
\node at (1.5, -0.5) {\Huge \bf \color{red} $\times$};

\draw (3, -1) -- (6, 0);
\draw (0, 0) -- (2, -2);
\draw [dashed] (2, -2) -- (4, -2);
\node at (3, -2) {\Huge \bf \color{red} $\times$};

\draw (4, -2) -- (6, 0);
\draw (0, 0) -- (1.5, -3) -- (3, -3) -- (4.5, -3) -- (6, 0);

\draw [line width = 0.5em, red] (0, 0) -- (6, 0);
\node [above] at (3, 0) {\color{red} \bf $4$};

\end{scope}

\end{tikzpicture}
\caption{\label{fig:weightupdate} (Left) In the original, pre-failure input graph, the thick red emulator edge is assigned weight $2$.  (Middle) Under one pair of failing edges, the weight of the emulator edge would update to $3$.  (Right) Under a different pair of failing edges, the weight of the emulator edge would update to $4$.}
\end{figure}

Why is this a reasonable model for fault-tolerant emulators?  It was introduced in~\cite{BDN22}, who gave both theoretical and practical justification for this model.  We refer the interested reader to~\cite{BDN22}, but also provide some brief justification.

\subsection{Theoretical Justification}
The main justification is that weight updates are the only natural option.  Suppose that we do not update edge weights, i.e., we fix the weight of an edge $\{u,v\}$ to something and we leave it at that value no matter what set of failures occur.  Then in order to satisfy the first inequality of emulators, the requirement that the distance between $u$ and $v$ in the emulator (post-failures) be at least the distance between $u$ and $v$ in $G$ (post-failures), we need to set this weight to be quite large.  For example, if there is a fault set which disconnects $u$ from $v$, then we have no choice but to set of the weight of $\{u,v\}$ to $\infty$.  This immediately implies that we cannot satisfy the second inequality of emulators, the stretch bound, for small (or empty) fault sets.

\subsection{Overlay Networks and Fault Tolerance}

The practical motivation for this definition is surprisingly strong.  Since we do not specify how automatic edge updates might actually be achieved, this definition might seem like ``cheating'' from a practical point of view.  But, in fact, the opposite is true: such edge updates can be thought of as abstracting away a process which \emph{should} be abstracted away!

While emulators are used in many different applications, they arise naturally in \emph{overlay networks}.  An overlay network is a logical network that lives ``on top of'' another network (usually the Internet): the vertices represent real nodes, but the edges are logical rather than physical, and sending data across an overlay edge is implemented by sending the data from one endpoint to the other across the underlying network.  This allows a network operator to abstract away the complications of the Internet (or other underlying network), and instead act as if they are operating on a simpler network.  It also allows for more innovative and interesting algorithms and protocols: at this point in time it is essentially impossible to change the structure and protocols of the Internet, so networking innovation often happens on top of the Internet on overlay networks.  Overlay networks are extremely useful (even though they simply run on top of the Internet), and have been extensively studied, often either directly or indirectly using spanners, emulators, or related objects (e.g., \cite{BWDA17,RON,ADGHT06,Detour,OverQos}).

Since a logical overlay edge is implemented through an underlying network, the distance that packets will travel when using such an edge is determined by the routing algorithm of the underlying network.  And the vast majority of routing algorithms implement shortest paths.  So each overlay edge actually represents a shortest path in the underlying network between the endpoints.  See, for example, Figure~\ref{fig:overlay}.

Thus after faults $F$, packets that are send along the $\{u,v\}$ logical edge will actually be routed using the underlying routing algorithm, which will be a shortest path between the endpoint in the underlying graph post-faults.   That is, the length of the logical edge $\{u,v\}$ after faults $F$ will automatically become $\dist_{G \setminus F}(u,v)$, thanks to the underlying routing algorithm reconverging to a shortest path.  So our ``automatic edge updates'' are a completely reasonable model for overlay networks.  And, of course, emulators are a natural fit for overlay networks: we want our overlay network to be simple (i.e., sparse), but with the property that when we actually use our overlay (abstracting away the underlying graph), distances (i.e., latencies) do not increase too much.

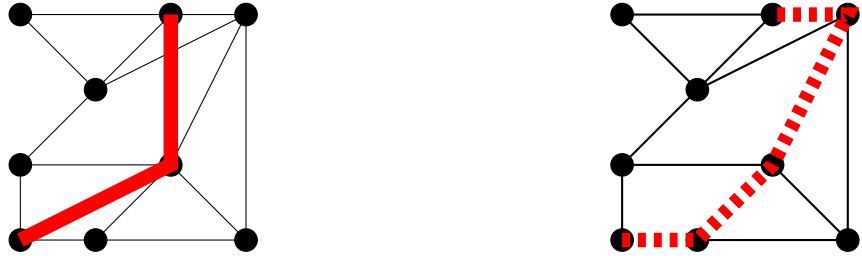
\begin{figure}[ht] \centering
\begin{tikzpicture}
\draw [fill=black] (0, 0) circle [radius=0.15];
\draw [fill=black] (1, 0) circle [radius=0.15];
\draw [fill=black] (0, 1) circle [radius=0.15];
\draw [fill=black] (2, 1) circle [radius=0.15];
\draw [fill=black] (3, 3) circle [radius=0.15];
\draw [fill=black] (2, 3) circle [radius=0.15];
\draw [fill=black] (0, 3) circle [radius=0.15];
\draw [fill=black] (1, 2) circle [radius=0.15];
\draw [fill=black] (3, 0) circle [radius=0.15];

\draw (0, 0) -- (1, 0);
\draw (0, 0) -- (0, 1);
\draw (0, 1) -- (1, 2);
\draw (1, 2) -- (0, 3);
\draw (1, 2) -- (3, 3);
\draw (1, 2) -- (2, 3);
\draw (3, 3) -- (2, 3);
\draw (0, 3) -- (2, 3);
\draw (2, 1) -- (3, 3);
\draw (2, 1) -- (3, 0);
\draw (1, 0) -- (3, 0);
\draw (1, 0) -- (2, 1);
\draw (0, 1) -- (2, 1);
\draw (3, 0) -- (3, 3);

\draw [line width = 0.5em, red] (0, 0) -- (2, 1) -- (2, 3);

\begin{scope}[shift={(8, 0)}]

\draw [fill=black] (0, 0) circle [radius=0.15];
\draw [fill=black] (1, 0) circle [radius=0.15];
\draw [fill=black] (0, 1) circle [radius=0.15];
\draw [fill=black] (2, 1) circle [radius=0.15];
\draw [fill=black] (3, 3) circle [radius=0.15];
\draw [fill=black] (2, 3) circle [radius=0.15];
\draw [fill=black] (0, 3) circle [radius=0.15];
\draw [fill=black] (1, 2) circle [radius=0.15];
\draw [fill=black] (3, 0) circle [radius=0.15];

\draw [thick] (0, 0) -- (0, 1);
\draw [thick] (0, 1) -- (1, 2);
\draw [thick] (1, 2) -- (0, 3);
\draw [thick] (1, 2) -- (3, 3);
\draw [thick] (1, 2) -- (2, 3);
\draw [thick] (0, 3) -- (2, 3);
\draw [thick] (2, 1) -- (3, 0);
\draw [thick] (1, 0) -- (3, 0);
\draw [thick] (0, 1) -- (2, 1);
\draw [thick] (3, 0) -- (3, 3);

\draw [line width = 0.5em, dashed, red] (0, 0) -- (1, 0) -- (2, 1) -- (3, 3) -- (2, 3);

\end{scope}
\end{tikzpicture}
\caption{\label{fig:overlay} (Left) An underlying network and a path of two edges from the overlay network on top of it.  (Right) The overlay network path will be automatically resolved into a path of physical links in the base graph, by resolving each edge into a shortest path between its endpoints.}
\end{figure}

\end{document}